\documentclass[12pt,reqno]{amsart}

\usepackage{hyperref}
\hypersetup{
colorlinks=true,
linkcolor={blue}
}

\usepackage[font=small,skip=0pt]{caption}
\usepackage{amsmath}
\usepackage{amsthm}
\usepackage{amssymb}
\usepackage{graphicx}
\graphicspath{{Figures_Var/}}

\newtheorem{theorem}{Theorem}
\newtheorem{definition}{Definition}
\newtheorem{lemma}{Lemma}
\newtheorem{remark}{Remark}

\newcommand{\R}{\mathbb{R}}

\begin{document}

\title[Solitary waves with intensity-dependent dispersion]{\bf Solitary waves with intensity-dependent dispersion: \\ variational characterization}

\author{D.E. Pelinovsky}
\address{Department of Mathematics and Statistics, McMaster
	University, Hamilton, Ontario, Canada, L8S 4K1}

\author{R.M. Ross}
\address{Department of Mathematics and Statistics,
         University of Massachusetts, Amherst, MA 01003-4515, USA}
     
\author{P.G. Kevrekidis}
\address{Department of Mathematics and Statistics,
 University of Massachusetts, Amherst, MA 01003-4515, USA}

\date{\today}

\begin{abstract}
A continuous family of singular solitary waves exists in a prototypical system with intensity-dependent dispersion. The family has a cusped soliton as the limiting lowest energy state and is formed by the solitary waves with bell-shaped heads of different lengths. We show that this family can be obtained variationally by minimization of mass at fixed energy and fixed length of the bell-shaped head. We develop a weak formulation for the singular solitary waves and prove that they are stable under perturbations which do not change the length of the bell-shaped head. Numerical simulations confirm the stability of the singular solitary waves. 
\end{abstract}

\maketitle

\section{Introduction}

Dispersive nonlinear systems typically feature the interplay of
dispersion
and nonlinearity that is prototypically represented through the
well-known model of the nonlinear Schr{\"o}dinger (NLS)
equation~\cite{NLS1,NLS2}.
This interplay is responsible for the formation of smooth solitary waves in a wide class of dispersive nonlinear systems. Nevertheless, 
some physical systems feature intensity-dependent dispersion (IDD);
relevant examples include the femtosecond pulse propagation in quantum well
waveguides~\cite{koser}, the electromagnetically induced transparency in
coherently prepared multistate atoms~\cite{greentree}, and 
fiber-optics communication systems~\cite{OL2020}. Such features
have been discussed in the context of both photonic, and in phononic
(acoustic)
crystals~\cite{mankeltow} and have even been argued to arise 
at the quantum-mechanical level (between Fock states) in the work of~\cite{kirchmair}.

Out focus here is in
addressing the NLS model with the prototypical IDD 
introduced in the context of
nonlinear optics in~\cite{OL2020}:
\begin{align}
\label{nls}
i\psi_t + (1-|\psi|^2)\psi_{xx} = 0,
\end{align}
where $\psi = \psi(x,t)$ is the complex wave function. It was shown in \cite{OL2020} that the NLS equation (\ref{nls}) admits formally two conserved quantities:
\begin{equation}
\label{mass-energy}
Q(\psi) = -\int_{\mathbb{R}} \log|1-|\psi|^2| dx, \quad
E(\psi) = \int_{\mathbb{R}} |\psi_x|^2 dx.
\end{equation}
The two conserved quantities have the meaning of the mass and energy
of the optical system. The conserved quantities (\ref{mass-energy}) 
are defined in the subspace of $H^1$ functions given by
\begin{equation}
\label{energy-space}
X = \left\{ u \in H^1(\mathbb{R}) : \quad \left| \int_{\mathbb{R}} \log|1-|u|^2| dx \right| < \infty \right\},
\end{equation}
which is the energy space of the NLS equation (\ref{nls}).

Our previous work in \cite{RKP} was devoted to the solitary waves 
in the NLS-IDD equation (\ref{nls}). Solitary waves arise as the
standing wave solutions of the form $\psi(x,t) = e^{i \omega t} u(x)$
with real $\omega$ and $u(x)$ (without loss of generality)
satisfying formally the nonlinear
differential equation
\begin{align}
\label{ode}
\omega u = (1-u^2) u''(x)
\end{align}
subject to the decay to zero at infinity. Since the classical solutions 
to the differential equation (\ref{ode}) are singular at the points of
$x$ where $u(x) = \pm 1$, solitary waves have to be defined in a weak 
formulation.

In the present work, we revisit this problem and propose a 
weak formulation which enables us to establish a notion
of Lyapunov stability of the singular solitary waves 
in the NLS-IDD equation (\ref{nls}). We use direct
numerical simulations to corroborate the theoretical results.

Our presentation is structured as follows. In section 2, we
present the mathematical background, basic definitions of
the problem, and state the main theorems. In section 3, we
prove the main results, while in section 4,  we illustrate them 
with numerical simulations. Finally, in section 5, we briefly summarize our findings.

\section{Mathematical Setup and Main Results}

We start with the definition of weak solutions of the differential equation 
(\ref{ode}) which was introduced in \cite{RKP}.

\begin{definition}
	\label{def-weak-solution}
	We say that $u \in H^1(\mathbb{R})$ is a weak solution 
	of the differential equation (\ref{ode}) if it satisfies 
	the following equation 
	\begin{equation}
	\label{ode-weak}
	\omega \langle u, \varphi \rangle + \langle (1- u^2) u', \varphi' \rangle 
	- 2 \langle u (u')^2, \varphi \rangle = 0, \quad \mbox{\rm for every } \varphi \in H^1(\mathbb{R}),
	\end{equation}
	where $\langle \cdot, \cdot \rangle$ is the standard inner product in $L^2(\mathbb{R})$. 
\end{definition}

The weak solutions are obtained as parts of the smooth orbits of the second-order differential equation (\ref{ode}). The smooth orbits satisfy 
the first-order invariant in the form 
\begin{align}
\label{ode-energy}
\frac{1}{2} \left(\frac{du}{dx} \right)^2 + \frac{\omega}{2} \log|1-u^2| = C,
\end{align}
where the value of $C$ is constant along every smooth orbit. 
It was proven in \cite{RKP} that a continuous family of weak 
solutions exists for each $\omega > 0$. The family describes positive and single-humped solitary waves shown on Fig. \ref{bluesolns}. The phase portrait computed from the energy levels $C$ is shown on Fig. \ref{ppandV_b=1}.

\begin{figure}[hbt]
	\centering
	\includegraphics[width=4cm,height = 4cm]{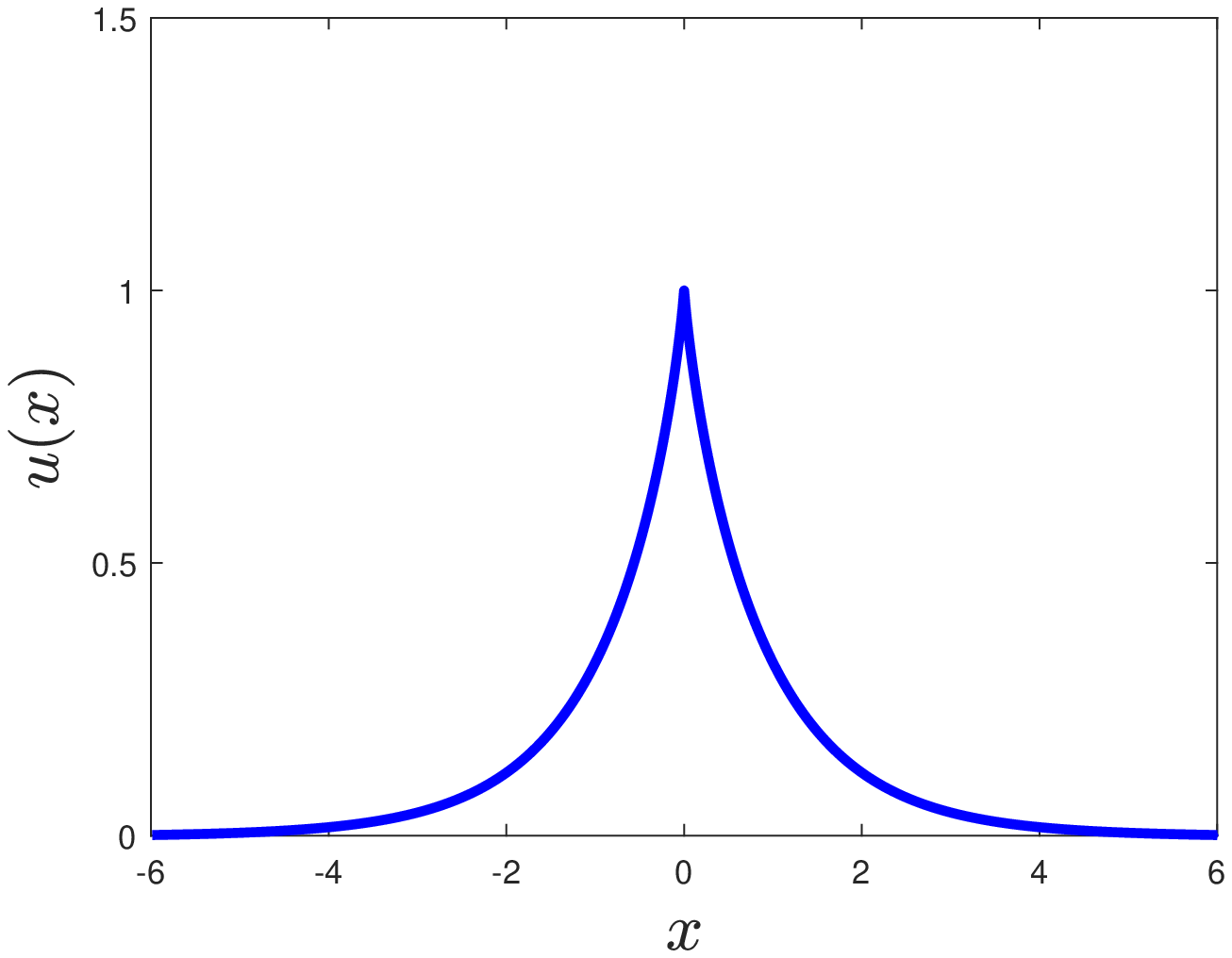}
	\includegraphics[width=4cm,height = 4cm]{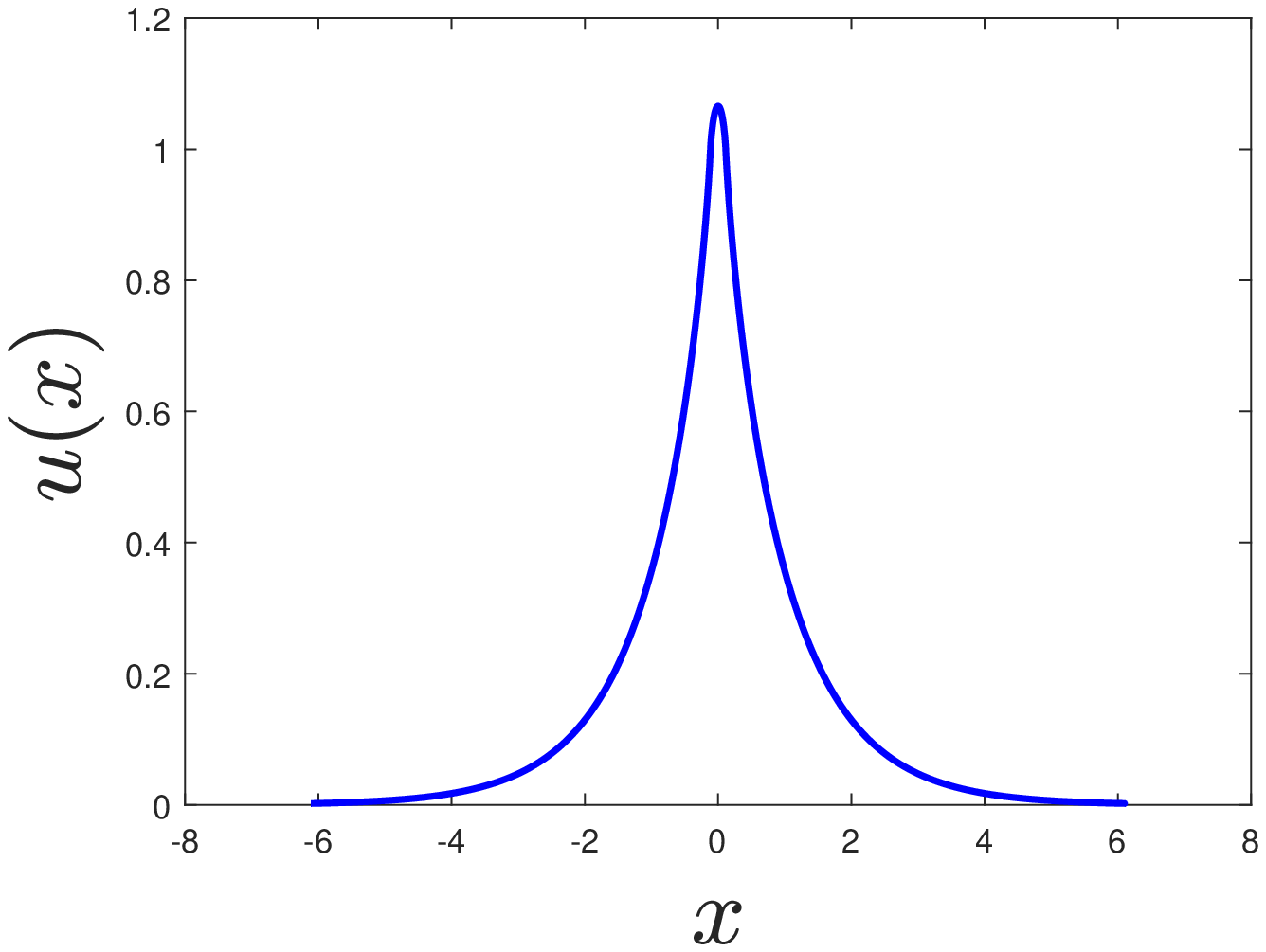} 
	\includegraphics[width=4cm,height = 4cm]{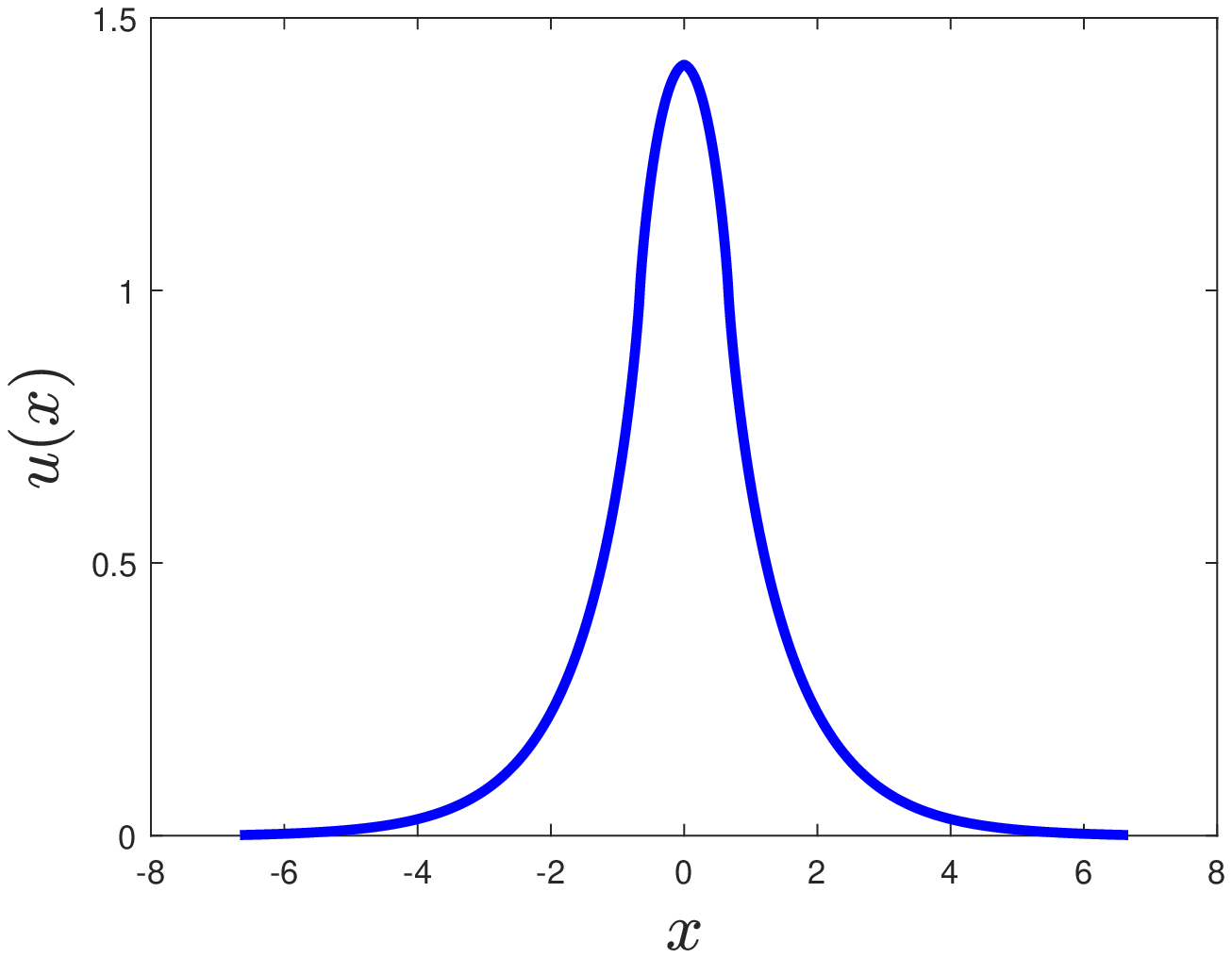}
	\includegraphics[width=4cm,height = 4cm]{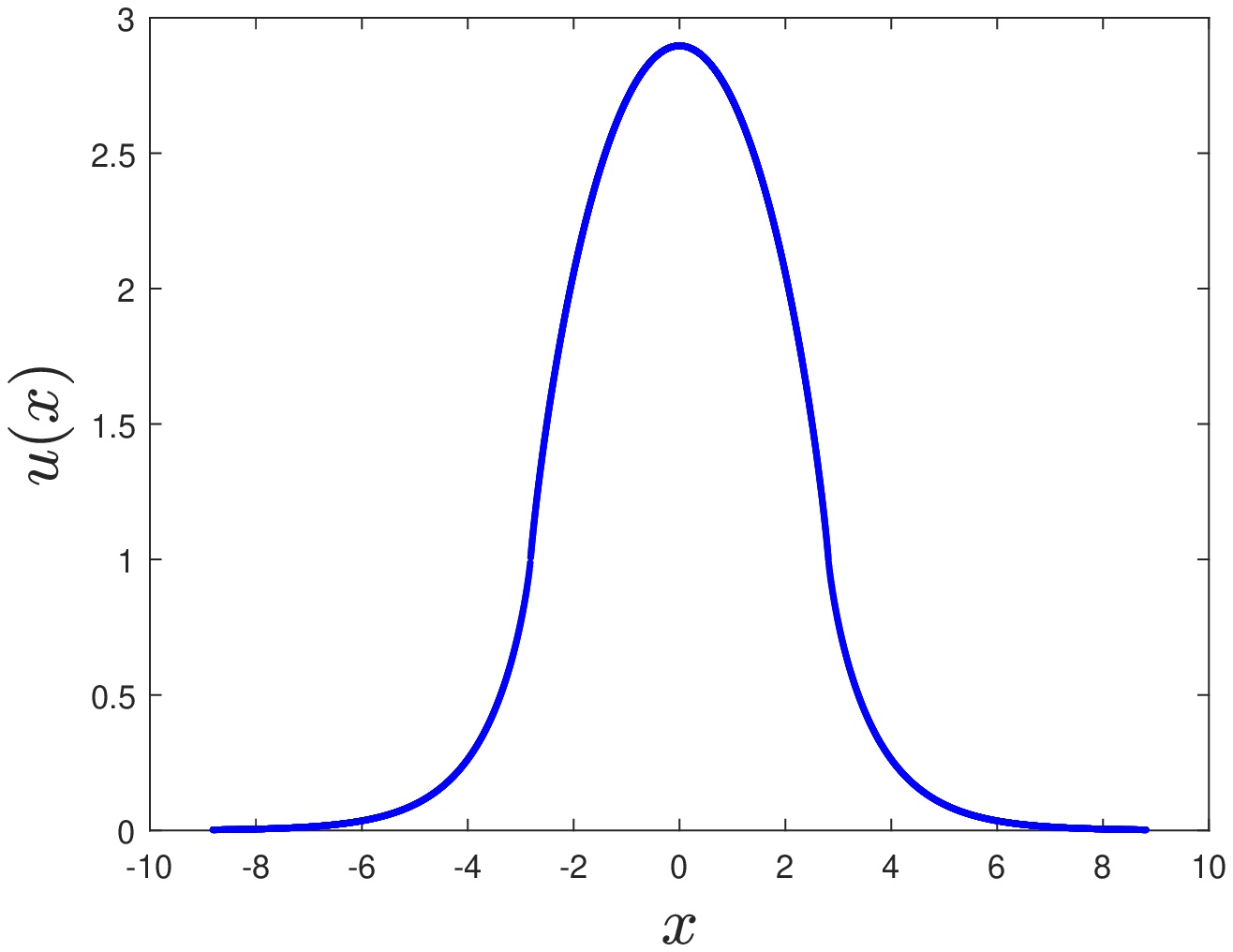}  
	\caption{The spatial profiles of four single-humped solitary wave solutions of the second-order equation (\ref{ode}) for $\omega = 1$. From left to right: $C=-\infty$ (cusped soliton), $C=-1$, $C=0$, and $C=1$.}
	\label{bluesolns}
\end{figure}

\begin{figure}[hbt]
	\centering
	\includegraphics[width=12cm,height=9cm]{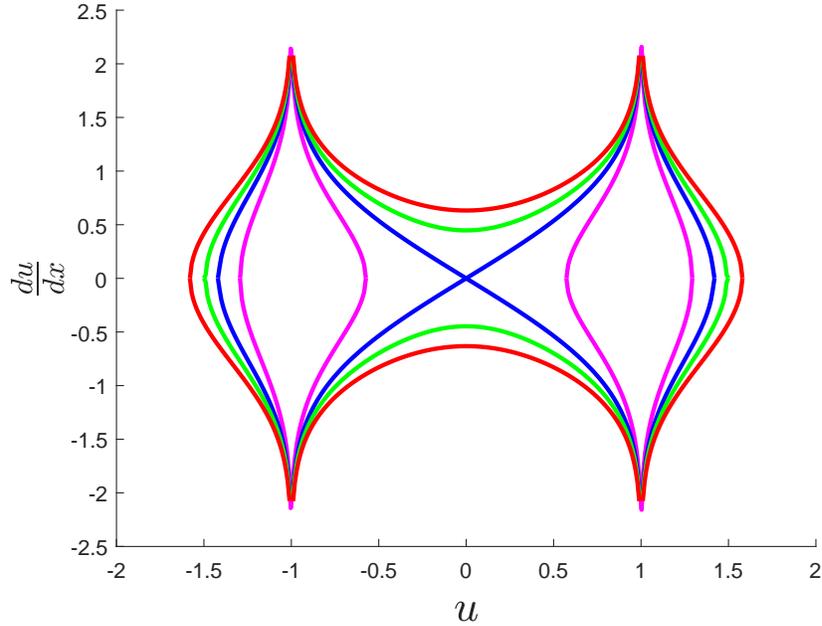}
	\caption{The phase portrait computed from the level curves (\ref{ode-energy}).}
	\label{ppandV_b=1}
\end{figure}

The exponentially decaying tails of the solitary waves correspond to 
the level $C = 0$, whereas the head 
of the bell-shaped solitary waves corresponds to an arbitrarily fixed value of  
$C$. The limiting cusped soliton satisfying $0 < u(x) \leq 1$ represents the lowest energy state in the family and corresponds formally to the 
limit $C \to -\infty$. The bell-shaped solitary 
wave for different values of $C \in \mathbb{R}$ has the head located in the interval $[-\ell_C,\ell_C]$, where $1 < u(x) \leq \sqrt{1+e^{2C}}$ 
for $x \in (-\ell_C,\ell_C)$. The tails and the head of the bell-shaped solitary waves are connected at the points 
$x = \pm \ell_C$, where $u(\pm \ell_C) = 1$ and $u'(\pm \ell_C)$ diverge. With the precise analysis of the asymptotic behavior of the solutions near the singularities (similar to \cite{Alfimov}), it was proven that 
the solutions satisfy the weak formulation in Definition \ref{def-weak-solution} 
and belong to the energy space $X$. The following theorem gives the summary 
of results obtained in \cite{RKP} under the normalization $\omega = 1$.

\begin{theorem}
	\label{theorem-main1}
	Fix $\omega = 1$. There exists a continuous family 
	of weak, positive, and single-humped solutions of Definition 
	\ref{def-weak-solution} parametrized by $C \in \mathbb{R}$ such that 
	\begin{equation}
	\label{cont-family}
	u_C(x) = \left\{ \begin{array}{ll} u_{{\rm head},C}(x), \quad & x \in [-\ell_C,\ell_C], \\
	u_{\rm cusp}(|x| - \ell_C), \quad & |x| > \ell_C, \end{array} \right.
	\end{equation}
	where $\ell_C$ is uniquely defined by
	\begin{equation}
	\label{def-L-C}
	\ell_C := \int_1^{\sqrt{1+e^{2C}}} \frac{du}{\sqrt{2C - \log(u^2-1)}},
	\end{equation}
	$u_{{\rm head},C}(x)$ for $x \in [-\ell_C,\ell_C]$ is defined
	implicitly by
	\begin{equation}
	\ell_C-|x| = \int_1^u \frac{d\xi}{\sqrt{2C-\log(\xi^2-1)}}, \quad u \in (1,\sqrt{1+e^{2C}}],
	\label{bell-head-continuous}
	\end{equation}
	and $u_{\rm cusp}(x)$ for $x \in \mathbb{R}$ is defined implicitly by
	\begin{align}
	|x| = \int_{u}^1 \frac{d\xi}{\sqrt{-\log(1-\xi^2)}}, \quad u \in (0,1).
	\label{0th}
	\end{align}
	Moreover, $u_C \in X \subset H^1(\mathbb{R})$ with the following 
	singular behavior as $|x| \to \ell_C$:
	\begin{align}
\label{asympt-sol-bell-cont}
	u_C(x) = 1 + (\ell_C-|x|) \sqrt{|\log|\ell_C-|x|||} \left[ 1 + \mathcal{O}\left(\frac{\log|\log|\ell_C-|x|||}{|\log|\ell_C-|x|||}\right)
\right],
	\end{align} 
	where $\mathcal{O}(v)$ denotes a $C^1$ function of $v$ at either side of $v = 0$.
\end{theorem}

The purpose of this work is to develop the variational characterization of the solitary wave solutions of Theorem \ref{theorem-main1} in order to prove their Lyapunov stability with respect to small perturbations. In order to place the solutions in the variational context and to deal with the singularity of the solitary wave solutions, we have to use a new definition of weak solutions.

\begin{definition}
	\label{def-weak-solution-new}
	Fix $L > 0$ and define 
	\begin{equation}
	\label{X-minimizers}
	X_L := \left\{ u \in X : \quad  u(x) > 1, \quad x \in (-L,L) \quad \mbox{\rm and}  \quad u(x) \leq 1, \quad |x| \geq L \right\}.
	\end{equation}
	Pick $u_L \in X_L$ satisfying
	$$
	\lim_{|x| \to L} \frac{u_L(x) - 1}{(L-|x|) \sqrt{|\log|L-|x|||}} = 1.
	$$
	We say that $u \in X_L \subset H^1(\mathbb{R})$ is a weak solution 
	of the differential equation (\ref{ode}) if it satisfies 
	the following equation 
	\begin{equation}
	\label{ode-weak-new}
	\langle u', \varphi' \rangle + \omega \langle (1-u^2)^{-1} u, \varphi \rangle = 0, \quad \mbox{\rm for every } \varphi \in H^1_L,
	\end{equation}
	where $H^1_L := \left\{ \varphi \in H^1(\mathbb{R}) : \;\; (1-u_L^2)^{-1} \varphi \in L^2(\mathbb{R}) \cap L^{\infty}(\mathbb{R}) \right\}$.
\end{definition}

The standard way to characterize smooth solitary waves in the NLS equation is 
to look for minimizers of energy $E(u)$ at fixed mass $Q(u)$ \cite{W86}. However, Lemmas \ref{lemma-1}, \ref{lemma-1-energy}, and \ref{lemma-1-mass} (proven by using methods developed in \cite{KMPX,NP}) show that the mappings 
$C \mapsto \ell_C$ and $C \mapsto E(u_C)$ are monotone whereas the mapping 
$C \mapsto Q(u_C)$ is non-monotone. As a result, {\em we develop a novel variational characterization of the singular solitary waves by looking at minimizers of mass $Q(u)$ at fixed energy $E(u)$.} 

The stationary equation (\ref{ode}) is the Euler--Lagrange equation for the action functional
\begin{equation}
\label{action}
\Lambda_{\omega}(u) = Q(u) + \omega^{-1} E(u),
\end{equation}
where $Q(u)$ and $E(u)$ are the conserved mass and energy of the NLS equation (\ref{nls}) given by (\ref{mass-energy}). Expanding formally 
\begin{equation}
\Lambda_{\omega}(u + \varphi) - \Lambda_{\omega}(u) = 2 \langle (1-u^2)^{-1} u, \varphi \rangle + 2 \omega^{-1} \langle u',\varphi' \rangle + 
\mathcal{O}(\| \varphi' \|^2_{L^2} + \| (1-u^2)^{-1} \varphi \|_{L^2\cap L^{\infty}}^2) 
\end{equation}
for $u \in X_L$ and $\varphi \in H^1_L$ yields that the weak solution of Definition \ref{def-weak-solution-new} is a critical point of $\Lambda_{\omega}(u)$ in $X_L$. Thus, we can consider the constrained minimization problem
\begin{equation}
\label{minimizers}
\mathcal{Q}_{\mu,L} := \inf_{u \in X_L} \{ Q(u) : \quad E(u) = \mu \}.
\end{equation}
In the context of the variational problem (\ref{minimizers}),
the parameter $\omega > 0$ serves as the Lagrange multiplier 
and the parameter $L > 0$ defines the support of the head of the solitary wave 
as in Definition \ref{def-weak-solution-new}. The solitary wave is weakly singular at $x = \pm L$.

The following theorem formulates the main result of the paper. 
The proof of this theorem relies on the monotonicity of the mappings 
$C \mapsto \ell_C$ and $C \mapsto E(u_C)$ (Lemmas \ref{lemma-1} and \ref{lemma-1-energy}), an elementary scaling argument (Lemma \ref{lemma-2}), and convexity of the second variation of the action functional $\Lambda_{\omega}(u)$ (Lemma \ref{lemma-3}).

\begin{theorem}
	\label{theorem-main2}
	For every $\mu > 0$ and $L > 0$, there exists a ground state, 
	that is, the minimizer of the constrained variational problem (\ref{minimizers}) in $X_L$. The minimizer coincides with a rescaled version of $u_C$ in Theorem \ref{theorem-main1} for some $C = C_{\mu,L}$. 
\end{theorem}

\begin{remark}
The result of Theorem \ref{theorem-main2} implies the 
Lyapunov stability of the solitary wave solutions of Theorem
\ref{theorem-main1} under perturbations which do not change the length
of the bell-shaped head. Stability of the solitary waves is confirmed
in the numerical simulations of the time-dependent NLS equation
(\ref{nls}) reported in section~\ref{sec-3}.
\end{remark}

\begin{remark}
The cusped soliton $u_{\rm cusp}$ can be included in the statement of Theorem \ref{theorem-main2} in the formal limit $L \to 0$ 
and $C_{\mu,L} \to -\infty$. It is also a minimizer of mass at fixed $E$ in the class of functions 
\begin{equation}
\label{class-functions-X0}
X_0 := \{ u \in X : \quad u(0) = 1 \;\; {\rm and} \;\; u(x) < 1, \;\; x \neq 0 \}. 
\end{equation}
This implies Lyapunov stability of the cusped soliton under perturbations in $X_0$.
\end{remark}

\begin{remark}
The methods and results obtained in this work are similar to the recent studies of compactons 
in the degenerate NLS equation \cite{GHM} and in the sublinear KdV equation 
\cite{PelinPelin}.
\end{remark}

\section{Proof of Theorem \ref{theorem-main2}}
\label{sec-2}

The following three lemmas address monotonicity of the mappings 
$C \mapsto \ell_C$, $C \mapsto E(u_C)$, and $C \mapsto Q(u_C)$. 
The proofs are based on the following standard
property from vector calculus.
If $W(u,v)$ is a $C^1$ function in an open region of $\mathbb{R}^2$, then the differential of $W$ is defined by 
$$
dW(u,v) = \frac{\partial W}{\partial u} du + \frac{\partial W}{\partial v} dv
$$
and the line integral of $d W(u,v)$ along any $C^1$ contour $\gamma$ connecting two points $(u_0,v_0)$ and $(u_1,v_1)$ does not depend on $\gamma$ and is evaluated as 
$$
\int_{\gamma} d W(u,v) = W(u_1,v_1) - W(u_0,v_0).
$$
A similar study of the monotonicity of the period function in the context of  differential equations on quantum graphs was recently performed in \cite{KMPX,NP}.

\begin{lemma}
	\label{lemma-1}
	Fix $\omega = 1$ and consider the solitary wave solutions of Theorem \ref{theorem-main1} parametrized by $C \in \mathbb{R}$. The mapping $C \mapsto \ell_C$ 
	is $C^1$ and monotonically increasing such that $\ell_C \to 0$ as $C \to -\infty$ and $\ell_C\to \infty$ as $C \to +\infty$.
\end{lemma}

\begin{proof}
	In order to show that the mapping $C \mapsto \ell_C$ 
	is $C^1$ and to compute $\frac{d\ell_C}{dC}$, we regularize the representation (\ref{def-L-C}) as follows 
\begin{eqnarray*}
	C \ell_C &=& \int_1^{\sqrt{1+e^{2C}}} \frac{C du}{\sqrt{2C - \log(u^2-1)}}\\ 
	&=& \frac{1}{2} \int_1^{\sqrt{1+e^{2C}}} 
	\left[ \sqrt{2C - \log(u^2-1)} + \frac{\log(u^2-1)}{\sqrt{2C- \log(u^2-1)}} \right] du,
\end{eqnarray*}
where the first-order invariant (\ref{ode-energy}) with $\omega = 1$ has been used. Denote $A(u) := \log(u^2-1)$ and write $v^2 + A(u) = 2C$ for the integral curve with the constant level $C$. Since $A'(u) \neq 0$ for $u > 1$, we have 
\begin{eqnarray*}
d \left[ \frac{A(u) v}{A'(u)} \right] &=& \left( 1 - \frac{A(u) A''(u)}{[A'(u)]^2} \right) v du + \frac{A(u)}{A'(u)} dv \\ 
&=& \left( 1 - \frac{A(u) A''(u)}{[A'(u)]^2} \right) v du - \frac{A(u)}{2v} du. 
\end{eqnarray*}
Therefore, the expression for $C \ell_C$ can be written in the non-singular form 
\begin{eqnarray*}
2 C \ell_C &=&  \int_1^{\sqrt{1+e^{2C}}} 
	\left[ 3 - \frac{2 A(u) A''(u)}{[A'(u)]^2} \right] v du \\
	 &=&  \int_1^{\sqrt{1+e^{2C}}} 
	\left[ 3 + \frac{1+u^2}{u^2} \log(u^2-1) \right] v du,
\end{eqnarray*}
where we have used that $v = 0$ at $u = \sqrt{1+e^{2C}}$ and 
$\lim_{u \to 1} (u^2-1) |\log(u^2-1)|^{3/2} = 0$. Since the right-hand side 
is a $C^1$ function of $C$, it follows that the mapping $C \mapsto \ell_C$ 
is $C^1$ so that differentiation in $C$ yields
\begin{eqnarray*}
	2 C \frac{d \ell_C}{dC} =  \int_1^{\sqrt{1+e^{2C}}} 
	\left[ 1 + \frac{1+u^2}{u^2} \log(u^2-1) \right] \frac{du}{v},
\end{eqnarray*}
where we have used that $1 = v \frac{\partial v}{\partial C}$ at fixed $u$. 
Let us now integrate by parts with the use of 
$$
\frac{d}{du} \left[ \frac{u^2-1}{u} \sqrt{2C- \log(u^2-1)} \right] = 
	-\frac{1}{\sqrt{2C- \log(u^2-1)}} + \frac{1+u^2}{u^2} \sqrt{2C- \log(u^2-1)}.
$$
Substituting it to the formula for $2C \frac{d \ell_C}{dC}$ and cancelling $2C$ on both sides of equation yields the final expression 
\begin{eqnarray}
\label{der-length}
\frac{d \ell_C}{dC} =  \int_1^{\sqrt{1+e^{2C}}} 
\frac{(1+u^2) du}{u^2 \sqrt{2C- \log(u^2-1)}},
\end{eqnarray}
which shows that $\frac{d \ell_C}{dC} > 0$. The limit $\ell_C\to 0$ as $C \to -\infty$ follows from the fact that both the integrand and the length of integration in (\ref{def-L-C}) converge to zero as $C \to -\infty$. 
On the other hand, the length of integration diverges as $e^C$ whereas the integrand converges to zero as $C^{-1/2}$ if $C \to \infty$, so that 
$\ell_C\to \infty$ as $C \to \infty$.
\end{proof}

\begin{lemma}
	\label{lemma-1-energy}
	In the setting of Lemma \ref{lemma-1}, the mapping $C \mapsto E(u_C)$ 
	is $C^1$ and monotonically increasing such that $E(u_C) \to E(u_{\rm cusp})$ as $C \to -\infty$ and $E(u_C) \to \infty$ as $C \to +\infty$.
\end{lemma}

\begin{proof}
	It follows from (\ref{cont-family}) that 
	$$
	E(u_C) = E(u_{\rm cusp}) + 2 \int_1^{\sqrt{1+e^{2C}}} \sqrt{2C- \log(u^2-1)} du,
	$$
	where the right-hand side is $C^1$ in $C$. Differentiating in $C$ yields
\begin{equation}
\label{der-energy}
	\frac{dE(u_C)}{dC} = 2 \int_1^{\sqrt{1+e^{2C}}} \frac{du}{\sqrt{2C- \log(u^2-1)}} = 2 \ell_C,
\end{equation}
	which shows that $\frac{d E(u_C)}{d C} > 0$. The length of integration 
	in the second integral for $E(u_C)$ converges to $0$ as $e^{2C}$ as $C \to -\infty$ whereas the integrand grows like $|C|^{1/2}$ as $C \to -\infty$. 
	Hence the second integral converges to $0$ and $E(u_C) \to E(u_{\rm cusp})$ as $C \to -\infty$. On the other hand, both the length of integration 
	and the integrand grow as $C \to +\infty$ so that $E(u_C) \to \infty$ as $C \to \infty$.
\end{proof}

\begin{lemma}
	\label{lemma-1-mass}
	In the setting of Lemma \ref{lemma-1}, the mapping $C \mapsto Q(u_C)$ 
	is $C^1$ and there exist $C_1 \leq C_2 < 0$ such that 
\begin{equation}
\label{mass-inequality}
	\frac{d}{dC} Q(u_C) > 0, \quad C \in (-\infty,C_1) \quad \mbox{\rm and} \quad
	\frac{d}{dC} Q(u_C) < 0, \quad C \in (C_2,\infty).
\end{equation}
\end{lemma}

\begin{proof}
	It follows from (\ref{cont-family}) that 
	$$
	Q(u_C) = Q(u_{\rm cusp}) - 2 \int_1^{\sqrt{1+e^{2C}}} \frac{\log(u^2-1)}{\sqrt{2C- \log(u^2-1)}} du,
	$$
	By using (\ref{ode-energy}), this expression can be rewritten as
	$$
	Q(u_C) = Q(u_{\rm cusp}) - 4 C \ell_C + E(u_C) - E(u_{\rm cusp}),
	$$
	where the right-hand side is $C^1$ in $C$ due to Lemmas \ref{lemma-1} and \ref{lemma-1-energy}. Differentiating in $C$ and using (\ref{der-energy}) yield
\begin{equation}
\label{der-mass}
	\frac{dQ(u_C)}{dC} =  - 2 \ell_C - 4 C \frac{d \ell_C}{dC}.
\end{equation}
It follows from positivity of (\ref{der-length}) that $\frac{d Q(u_C)}{d C} < 0$ for $C \geq 0$. On the other hand, since $\ell_C\to 0$ as $C \to -\infty$, 
positivity of (\ref{der-length}) implies that $\frac{d Q(u_C)}{d C} > 0$ for 
sufficiently large negative $C$. By continuity of $\frac{dQ(u_C)}{dC}$, there exist $C_1 \leq C_2 < 0$ such that the signs in (\ref{mass-inequality}) hold.
\end{proof}

\begin{remark}
Figure~\ref{fig-EandQ} shows all three mappings as functions of $C$. 
It suggests that the mapping $C \mapsto Q(u_C)$ has exactly one critical point, 
that is, $C_1 = C_2$ in the statement of Lemma \ref{lemma-1-mass}. We
were not able to prove this property from the analysis of (\ref{der-mass}), 
unlike the direct proof of Lemma 10 in \cite{NP}.
\end{remark}

\begin{figure}[hbt]
	\includegraphics[scale=0.35]{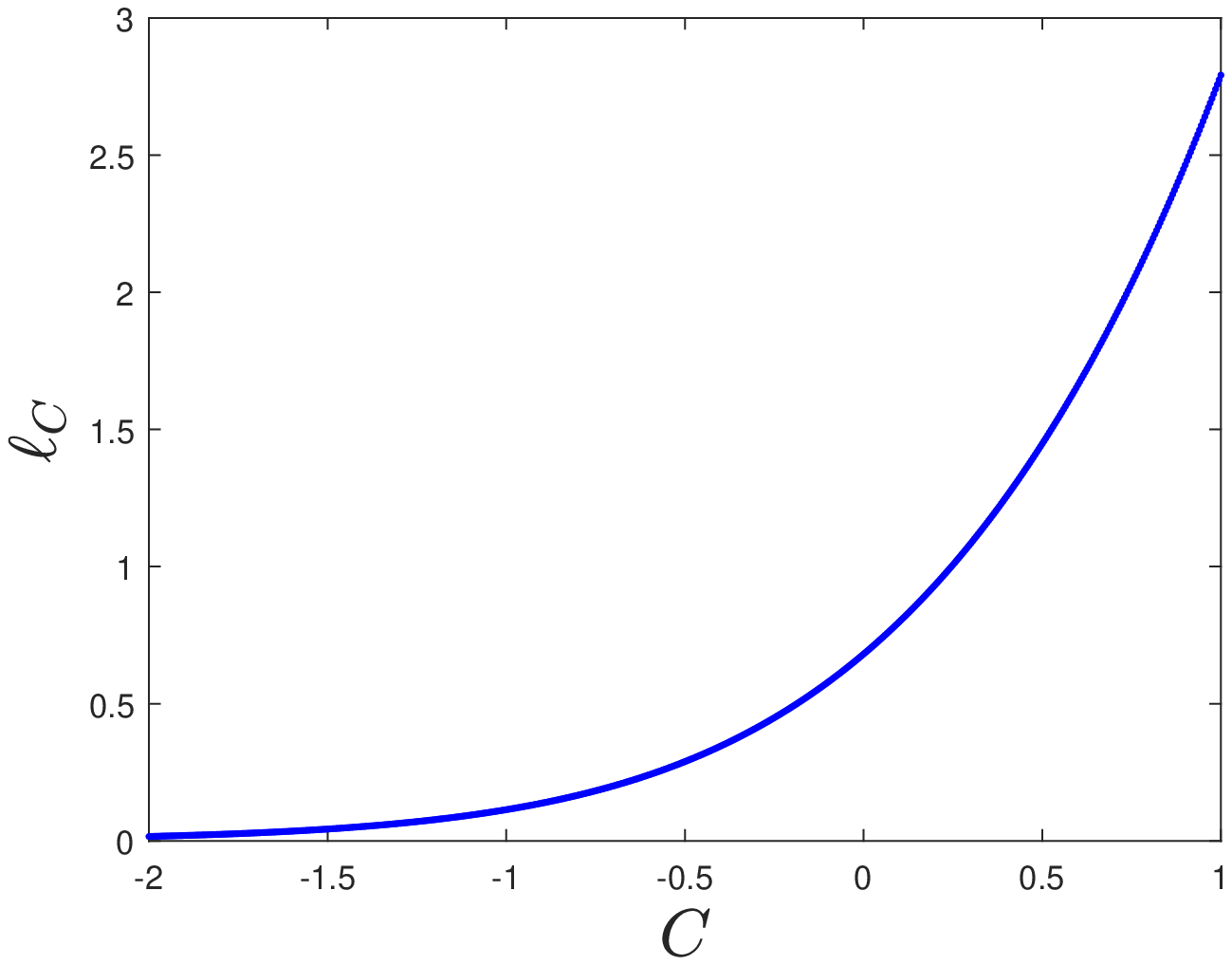}
	\includegraphics[scale=0.35]{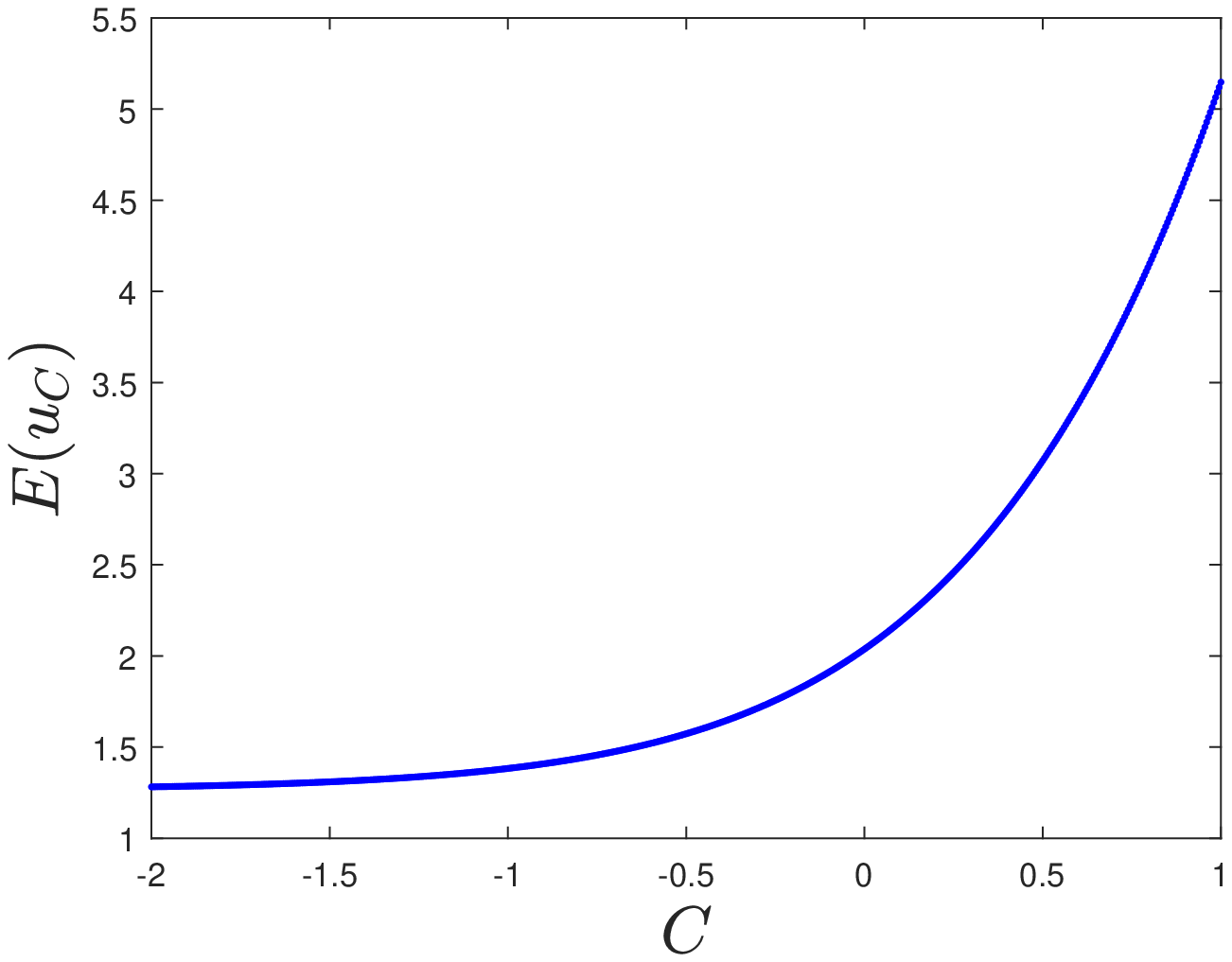}
	\includegraphics[scale=0.35]{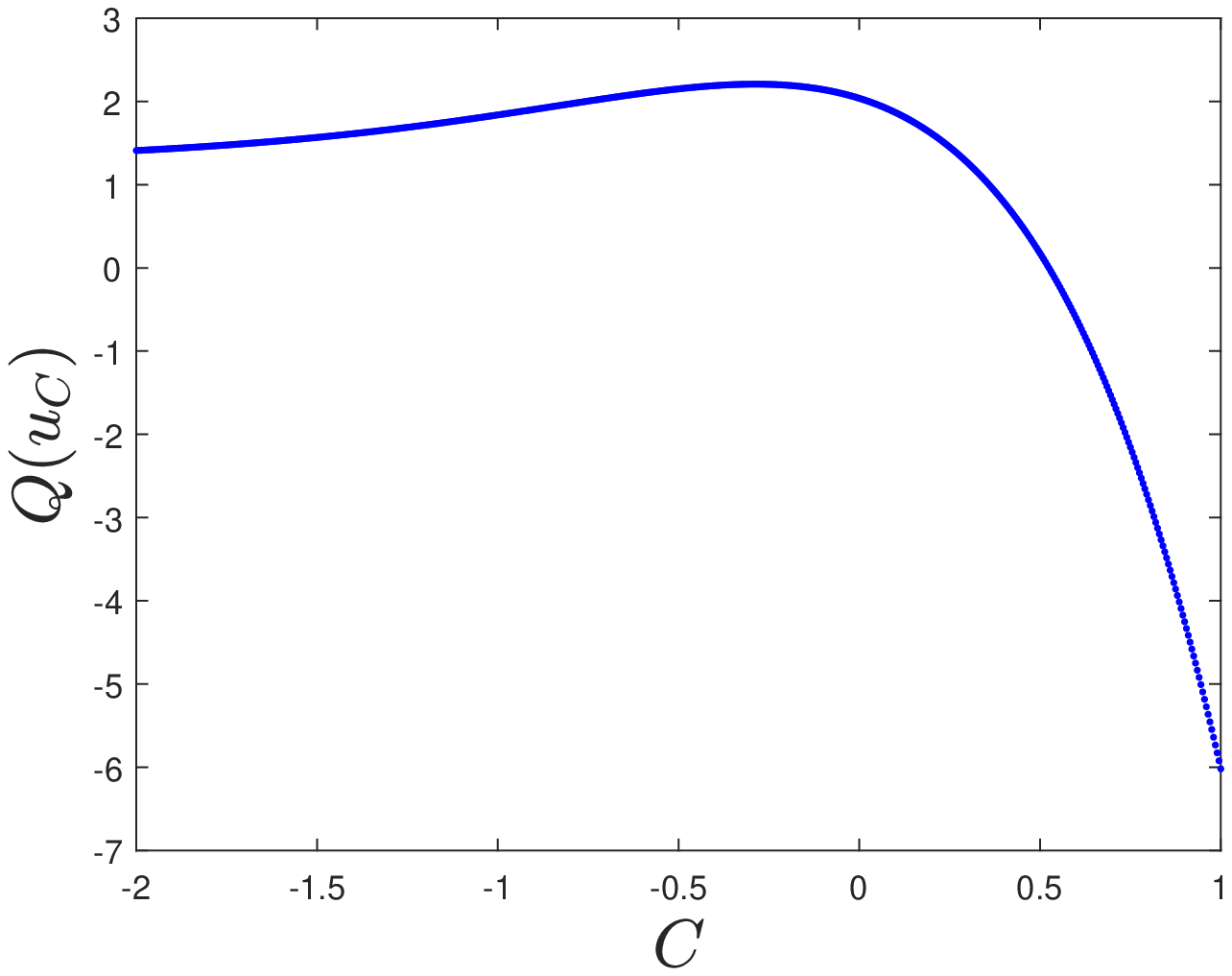}
	\caption{Dependencies of $\ell_C$ (left), $E(u_C)$ (middle), and $Q(u_C)$ (right) versus $C$. }
	\label{fig-EandQ}
\end{figure}

The following lemma uses the scaling transformation to obtain a critical point of the constrained variational problem (\ref{minimizers}).

\begin{lemma}
	\label{lemma-2}
For every $\mu > 0$ and $L > 0$, there exists a unique value of $C = C_{\mu,L}$ such that a critical point of the constrained variational problem (\ref{minimizers}) is defined by the solution $u_C$ of Theorem \ref{theorem-main1}.
\end{lemma}

\begin{proof}
	Let $u_C$ be the solitary wave solution of the normalized equation
	\begin{align}
	\label{2nd}
	 u'' = \frac{u}{1-u^2}.
	\end{align} 
	The scaled function $u_{\omega}(x) = u_C(\sqrt{\omega}x)$ is a solution of the second-order equation (\ref{ode}) for $\omega > 0$ and is the critical point 
	of the action functional $\Lambda_{\omega}(u)$ given by (\ref{action}) in $X_L$. 	Using the scaling transformation in the conserved mass and energy in (\ref{mass-energy}) gives
	\begin{equation*}
	Q(u_{\omega}) = \frac{1}{\sqrt{\omega}} Q(u_C), \quad E(u_{\omega}) = \sqrt{\omega} E(u_C).
	\end{equation*}
The singularities of $u_{\omega}$ are located at 
	\begin{equation*}
L = \frac{1}{\sqrt{\omega}} \ell_C.
\end{equation*}
	The Lagrange multiplier $\omega$ is selected from the condition 
	$\mu = E(u_{\omega}) = \sqrt{\omega} E(u_C)$. 
	Computing the Jacobian of the transformation 
	$(\omega,C) \mapsto (\mu,L)$ by 
	\begin{equation}
	\left| \begin{array}{cc} \displaystyle  \frac{\partial\mu}{\partial \omega} & \displaystyle\frac{\partial\mu}{\partial C} \\ \\
	\displaystyle \frac{\partial L}{\partial \omega} & \displaystyle \frac{\partial L}{\partial C} \end{array} \right| = 
	\frac{1}{2 \omega} \left[ E(u_C) \frac{d\ell_C}{d C} + \ell_C \frac{dE(u_C)}{d C} \right], 
	\end{equation}
	it follows by Lemmas \ref{lemma-1} and \ref{lemma-1-energy} that the Jacobian is positive for every $C \in \mathbb{R}$. Hence
	the mapping $(\omega,C) \mapsto (\mu,L)$ is invertible and there exists a unique $C = C_{\mu,L}$ for every $\mu > 0$ and $L > 0$. 
\end{proof}

\begin{remark}
	\label{remark-1}
	If $L = 0$, then $\ell_C = 0$ in the formal limit $C\to -\infty$. 
	This yields the cusped soliton $u_{\rm cusp}$ with $\sqrt{\omega} = \frac{\mu}{E(u_{\rm cusp})}$ for every $\mu > 0$. Hence the limiting value $L = 0$ can be included in the statement of Lemma \ref{lemma-2} with 
	$\lim_{L\to 0} C_{\mu,L} = -\infty$ for fixed $\mu > 0$. 
\end{remark}

\begin{remark}
		\label{remark-2}
	If $\mu = 0$, then $E(u_{\omega}) = 0$ with the only solution 
	$u_{\omega}(x)$ being a constant. If the constant is nonzero, then $u_{\omega} \notin  H^1(\mathbb{R})$. If the constant is zero, then 
	$L$ is not defined. In either case, the limiting value $\mu = 0$ cannot be included 
	in the statement of Lemma \ref{lemma-2}. 
\end{remark}

\begin{remark}
	The inverse transformation of the mapping $(\omega,C) \mapsto (\mu,L)$ in the proof of Lemma \ref{lemma-2} can be made explicit. Since $\sqrt{\omega} = \frac{\ell_C}{L}$ and $\mu = \sqrt{\omega} E(u_C)$, $C$ is uniquely found from the equation 	$\ell_C E(u_C) = L \mu$. Hence, $C_{\mu,L} \equiv C_{\mu L}$ depends on one parameter $\mu L$. This dependence is shown in Figure \ref{fig-C_muL}.
\end{remark}

\begin{figure}[hbt]
\includegraphics[scale=0.5]{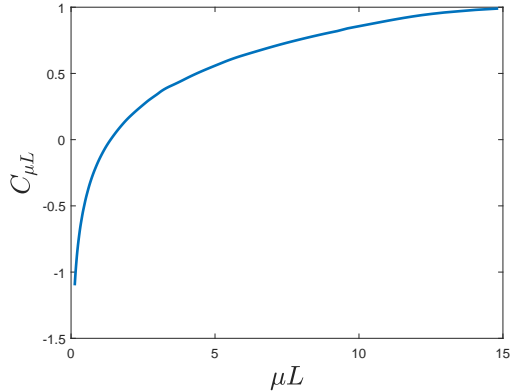}
\caption{Dependence of $C_{\mu L}$ from Lemma \ref{lemma-2} on the parameter $\mu L$.}
\label{fig-C_muL}
\end{figure}

The following lemma states that the critical point of Lemma \ref{lemma-2} is in fact a strict local minimizer of the constrained variational problem (\ref{minimizers}).

\begin{lemma}
	\label{lemma-3}
	Fix $\omega = 1$ and $C \in \mathbb{R}$. The solution $u_C$ of Theorem \ref{theorem-main1} is a strict local minimizer of the action functional 
	$\Lambda_{\omega = 1}(u)$ in $X_{L = \ell_C}$.
\end{lemma}

\begin{proof}
Let $u_C \in X_{\ell_C} \subset H^1(\mathbb{R})$ be a solitary wave solution of the normalized equation (\ref{2nd}). Let $v + i w$  with real $v,w \in H^1_{\ell_C} \subset H^1(\mathbb{R})$ be a perturbation to $u_C$. Since $(1-u_C^2)^{-1} v, 
(1-u_C^2)^{-1} w \in L^{\infty}(\mathbb{R})$ and $u_C(\pm \ell_C) = 1$, it follows that $v(\pm \ell_C) = w(\pm \ell_C) = 0$. 

Expanding $\Lambda_{\omega=1}(u_C+v+iw)$ in powers of $(v,w)$ and integrating 
by parts for $\int_{\mathbb{R}} u_C'(x) v'(x) dx$ with $v(\pm \ell_C) = 0$ 
and $(1-u_C^2)^{-1} v \in L^{\infty}(\mathbb{R})$ yields a vanishing linear term in $(v,w)$ because $u_C$ is the critical point of the action functional $\Lambda_{\omega=1}(u)$. Continuing the expansion to the quadratic and higher orders in $(v,w)$ yields the following expansion
\begin{equation}
\label{second-var}
\Lambda_{\omega=1}(u_C+v+iw) = \Lambda_{\omega=1}(u_C) + Q_+(v) + Q_-(w) + R(v,w),
\end{equation}
where $Q_{\pm}$ are the quadratic forms given by 
\begin{equation*}
\label{Lplus}
Q_+(v) = \int_{\mathbb{R}} \left[ (v_x)^2 + \frac{(1+u_C^2) v^2}{(1-u_C^2)^2} \right] dx
\end{equation*}
and
\begin{equation*}
\label{Lminus}
Q_-(w) = \int_{\mathbb{R}} \left[ (w_x)^2 + \frac{w^2}{1-u_C^2} \right] dx,
\end{equation*}
whereas $R(v,w)$ is the remainder term given by 
\begin{equation*}
\label{Rem}
R(v,w) = -\int_{\mathbb{R}} \left[ 
\log\left( 1 - \frac{2u_C v + v^2 + w^2}{1-u_C^2}\right) + \frac{2u_C v}{1-u_C^2} + \frac{(1+u_C^2) v^2}{(1-u_C^2)^2} + \frac{w^2}{1-u_C^2} \right] dx.
\end{equation*}
The quadratic forms $Q_{\pm}$ and the remainder term $R(v,w)$ are bounded 
since $v,w \in H^1(\mathbb{R})$ and $(1-u_C^2)^{-1} v, 
(1-u_C^2)^{-1} w \in L^2(\mathbb{R}) \cap L^{\infty}(\mathbb{R})$ are small for the perturbation terms $v,w$. In particular, by Taylor expansion of the logarithmic function, it follows that there exists a positive constant $C$ 
such that it is true for all small perturbation terms $v, w \in H^1_{\ell_C}$ that 
\begin{equation*}
|R(v,w)| \leq C \| (1-u_C^2)^{-1} v \|^3_{L^2 \cap L^{\infty}} + C \| (1-u_C^2)^{-1} w \|^3_{L^2 \cap L^{\infty}},
\end{equation*}
so that the remainder term is cubic with respect to the perturbation terms. 
We claim that there exist $C_{\pm} > 0$ such that 
\begin{equation}
\label{quad-form-coercive}
Q_{\pm}(v) \geq C_{\pm} \| v \|^2_{H^1}, \quad v \in H^1_{\ell_C} \subset H^1(\mathbb{R}),
\end{equation}
hence the quadratic forms are strictly positive and $u_C$ is a strict local minimizer of the action functional $\Lambda_{\omega = 1}(u)$ in $X_{L = \ell_C}$ by the second derivative test.

It remains to prove the bounds (\ref{quad-form-coercive}). Since $v(\pm \ell_C) = w(\pm \ell_C) = 0$, the domain $\mathbb{R}$ is partitioned to $(-\infty,-\ell_C) \cup (-\ell_C,\ell_C) \cup (\ell_C,\infty)$ and each 
quadratic form is considered separately in each interval 
subject to the Dirichlet boundary condition at $x = \pm \ell_C$.

Since $0 < u_C(x) \leq \sqrt{1+ e^{2C}}$, we have 
$$
\frac{1+u_C^2}{(1-u_C^2)^2} \geq \min(1,e^{-4C}),
$$
hence the bound (\ref{quad-form-coercive}) holds for $Q_+$ 
with $C_+ := \min(1,e^{-4C})$.

Since $(1-u_C^2)^{-1}$ is sign-indefinite, special treatment 
is needed for $Q_-$. 
On each interval of the partition $\mathbb{R} = (-\infty,-\ell_C) \cup (-\ell_C,\ell_C) \cup (\ell_C,\infty)$, 
the quadratic form $Q_-$ can be expressed in terms of the differential operator $L_-$ given by 
\begin{equation}
\label{operator-L-minus}
L_- = - \partial_x^2 + \frac{1}{1-u_C^2}. 
\end{equation}
The spectral problem for $L_-$ is set on $(-\infty,-\ell_C)$, $(-\ell_C,\ell_C)$, and $(\ell_C,\infty)$
subject to the Dirichlet conditions at $x = \pm \ell_C$. 
This defines the spectrum of $L_-$ in $L^2((-\infty,-\ell_C) \cup (-\ell_C,\ell_C) \cup (\ell_C,\infty))$ with the domain 
$H^2_0(-\infty,-\ell_C)\cap H^2_0(-\ell_C,\ell_C)\cap H^2_0(\ell_C,\infty)$.

On the other hand, $L_-$ can also be considered in $L^2(\mathbb{R})$ 
with a suitably defined domain in $L^2(\mathbb{R})$. Since $u_C(x) \to 0$ as $|x| \to \infty$ exponentially fast, Weyl's theorem implies that the essential spectrum of $L_-$ in $L^2(\mathbb{R})$ is located on $[1,\infty)$. 
Since $L_- u_C = 0$ with $u_C \in H^1(\mathbb{R})$ 
and $u_C(x) > 0$ for every $x \in \mathbb{R}$,
Sturm's theorem implies that the discrete spectrum of $L_-$ 
in $L^2(\mathbb{R})$ is located in $[0,1)$ and $0$ is a simple eigenvalue 
of $L_-$ in $L^2(\mathbb{R})$. 
 
When $L_-$ is restricted on $L^2((-\infty,-\ell_C) \cup (-\ell_C,\ell_C) \cup (\ell_C,\infty))$ subject to the Dirichlet 
conditions at $x = \pm \ell_C$, the smallest eigenvalue 
of $L_-$ becomes positive in $L^2((-\infty,-\ell_C) \cup (-\ell_C,\ell_C) \cup (\ell_C,\infty))$ because $u_C(\pm \ell_C) = 1 \neq 0$. Hence, the bound (\ref{quad-form-coercive}) holds for $Q_-$ with some $C_-$ given by 
the smallest eigenvalue of $L_-$ in $L^2((-\infty,-\ell_C) \cup (-\ell_C,\ell_C) \cup (\ell_C,\infty))$.
\end{proof}

\begin{remark} 
	\label{remark-cusped}
For the cusped soliton with $\ell_C= 0$ as $C \to -\infty$, the bounds
(\ref{quad-form-coercive}) hold with $C_{\pm} = 1$ since $0 < u(x)
\leq 1$ for all $x \in \mathbb{R}$. It is then not necessary to partition $\mathbb{R}$ into $(-\infty,0) \cap (0,\infty)$ for the proof of these bounds.
\end{remark}

We are now ready to prove Theorem \ref{theorem-main2}. 

By Lemma \ref{lemma-2}, for every $\mu > 0$ and $L > 0$, the critical point of the constrained variational problem 
(\ref{minimizers}) is given by $u_{\omega}(x) = u_C(\sqrt{\omega} x)$ with uniquely defined $C = C_{\mu L}$ and $\omega = \left[ \frac{\ell_{C_{\mu L}}}{L}\right]^2$. 
By Lemma \ref{lemma-3}, this critical point is a local minimizer of the action functional $\Lambda_{\omega}(u)$. From Theorem \ref{theorem-main1}, 
no other critical points satisfying the Euler--Lagrange equation 
(\ref{ode-weak-new}) exist in $X_L$. Therefore, the critical point 
is the global minimizer of mass $Q(u)$ for fixed energy $E(u) = \mu$. 
The proof extends to $\mu > 0$ and $L = 0$ with $u_{\rm cusp}$ replacing $u_C$ 
by Remarks \ref{remark-1} and \ref{remark-cusped}.

\section{Time evolution of perturbations}
\label{sec-3}

In order to corroborate the results regarding the existence and
stability of the
minimizer of the constrained variational problem, 
we investigate the time evolution of perturbations of the solitary
wave with the profile $u_C$ for some uniquely selected $C = C_{\mu
  L}$. We consider perturbations of the singular solitary waves which
do not alter the singularity location at $\pm L$ with $L = \ell_C$, 
in line with our theoretical analysis,
but change the energy level $E(u_C) = \mu$. We do this by perturbing
only
the head portion of the solitary wave on $(-\ell_C,\ell_C)$, while leaving the solution in the outer regions unchanged. The initial condition is given by 
\begin{align}
\label{pertIC}
u_P(x) = 
\begin{cases}
P u_{{\rm head}, C}(x), \quad & |x| < \ell_C
\\
u_{{\rm cusp}}(|x|-\ell_C) , \quad &|x|\geq \ell_C
\end{cases}
\end{align}
where the perturbation factor $P$ is close to $1$, both for $P>1$,
e.g., $P=1.1$, and for $P<1$, such as, e.g., $P=0.9$. 

To perform the time evolution of the NLS equation \eqref{nls}, we use a pseudospectral method. First, we discretize the interval $[-500,500]$ with $N=2000$ points. Next, spatial derivatives $\psi_x$ and $\psi_{xx}$ on the grid are approximated by vectors $D_1\psi$ and $D_2 \psi$ respectively, where $D_1$ and $D_2$ are matrix representations of the first and second derivative operators based on the circulant matrices from \cite{Tr-book}. 
Finally, time integration is performed using the fourth-order Runge-Kutta method, with time step $\Delta t = 0.001$.  

\begin{figure}[hbt]
	\includegraphics[scale=0.45]{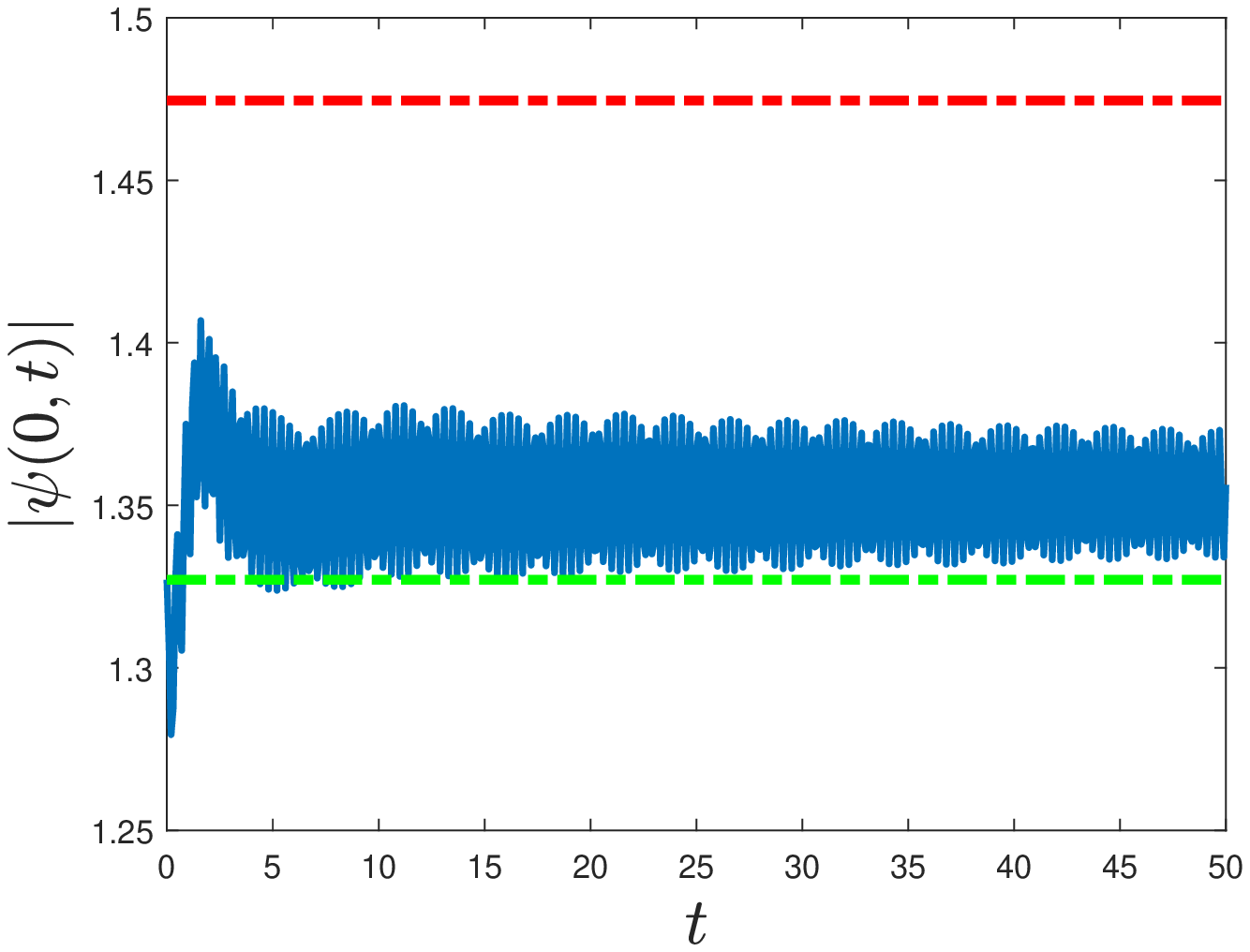}
	\includegraphics[scale=0.45]{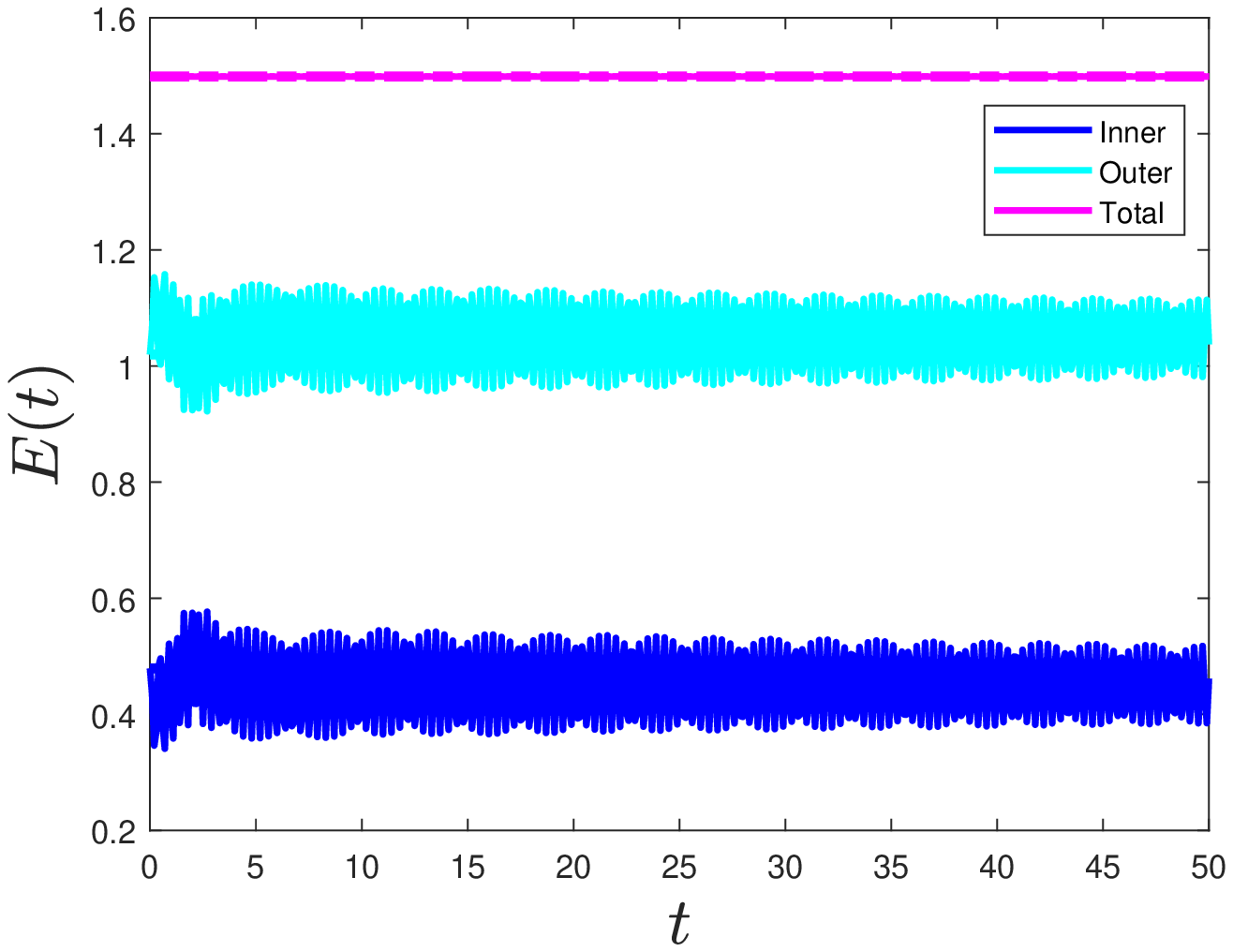}
	\\
	\includegraphics[scale=0.45]{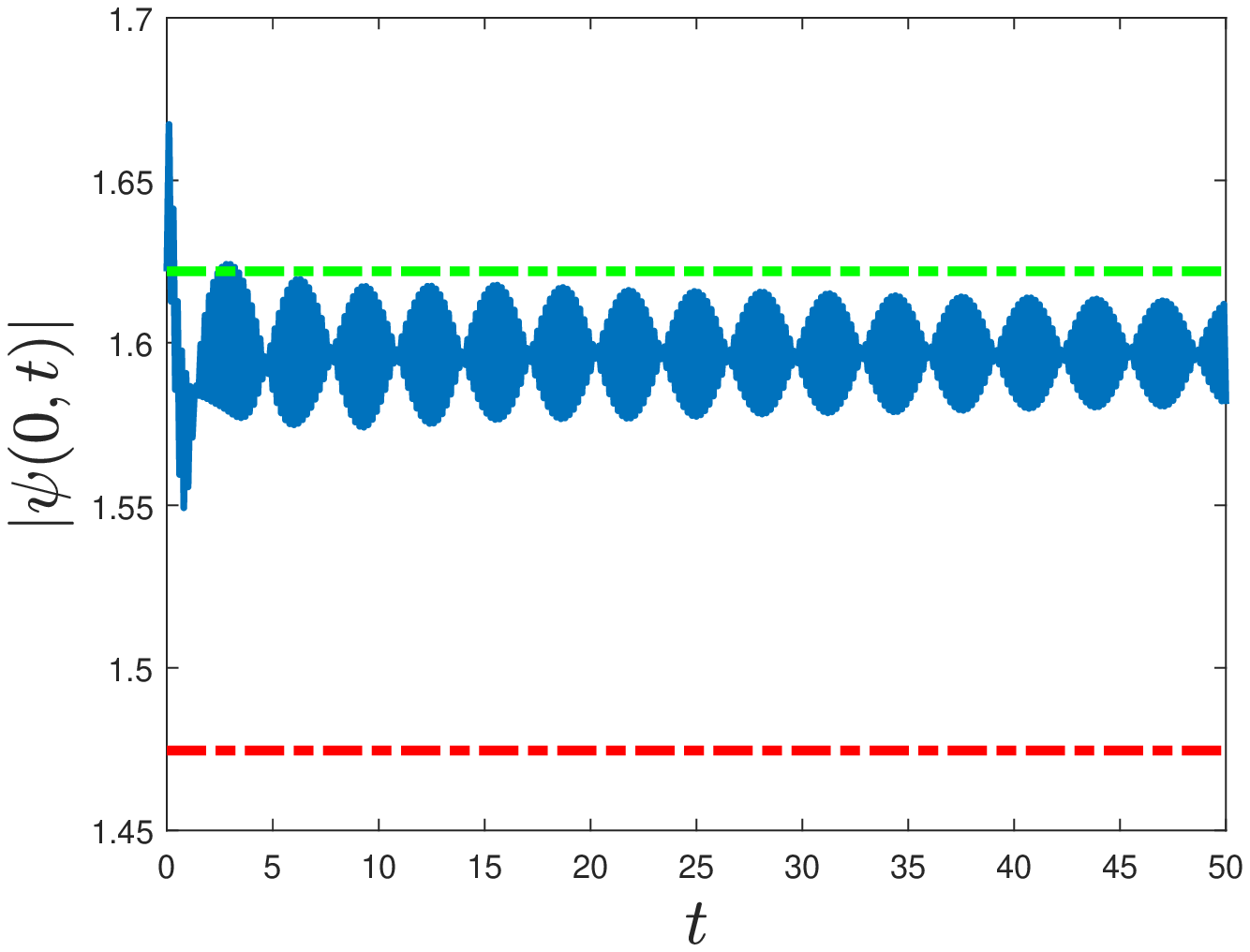}
	\includegraphics[scale=0.45]{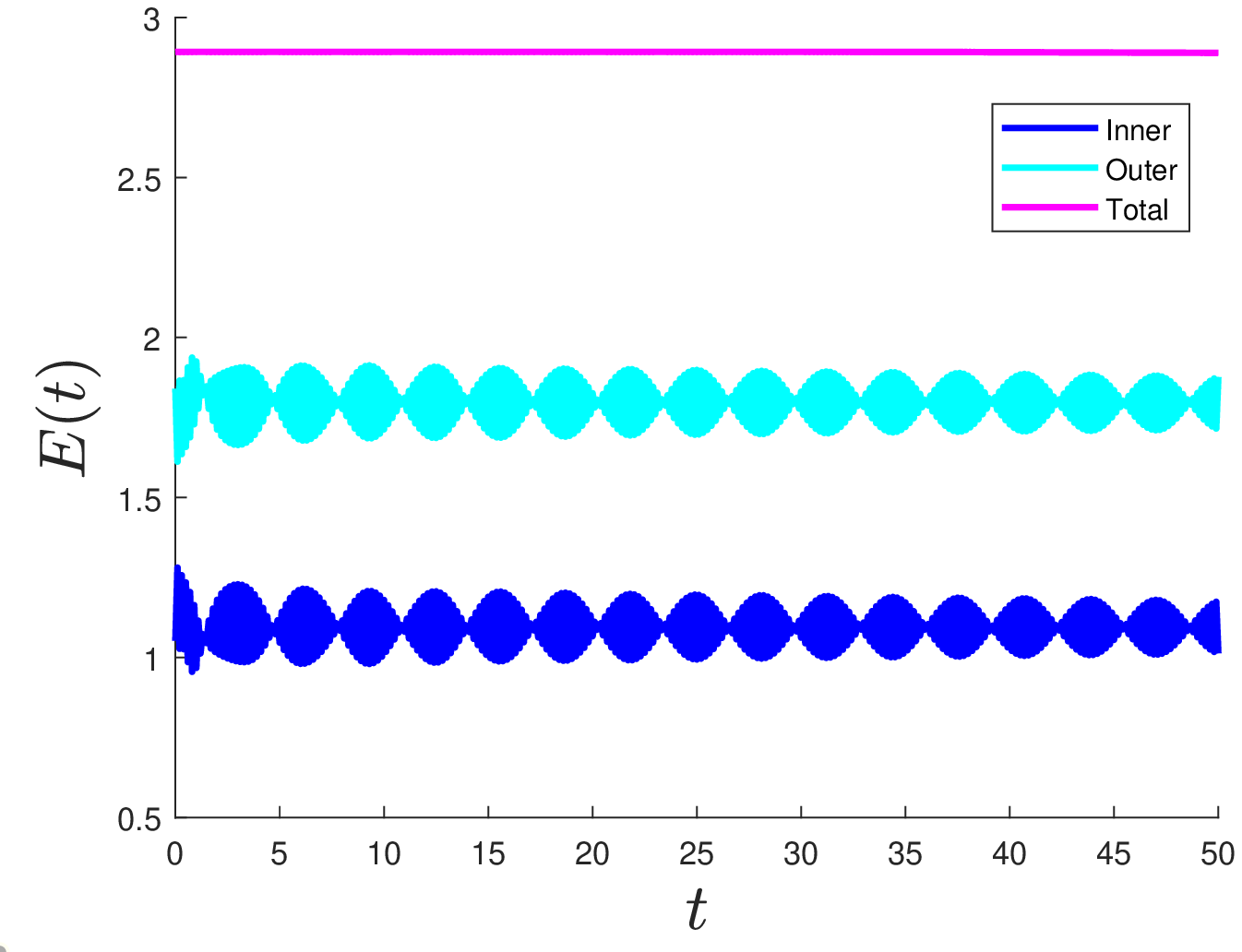}
	\caption{Amplitude (left) and ---inner and outer, as defined
          in the text, as well as
          total-- energies (right) for the time evolution of the initial condition \eqref{pertIC} with $C=0$ (the bell-shaped soliton), where $P=0.9$ (top) and $P=1.1$ (bottom). The green dotted line represents the amplitude of the initial condition, while the red one represents the amplitude of the unperturbed soliton.}
	\label{bell-pert}
\end{figure}

At each time $t$ in the evolution, we compute the energy contained in the inner and outer portions, given by
\begin{align}
E_{{\rm inner}}(t) &= \int_{-L}^{L} |\psi_x(x,t)|^2 \, dx
\\
E_{{\rm outer}}(t) &= \int_{|x|>L} |\psi_x(x,t)|^2 \, dx.
\end{align} 
This gives insight into how energy may be exchanged between the inner and outer regions.
We  also plot the value at the peak $|\psi(0,t)|$. Figures
\ref{bell-pert}--\ref{C=-05-pert} show these diagnostics for three
members of the solitary wave family (i.e., $C=0$, $C=0.5$ and
$C=-0.5$)
with two perturbation factors at
$P=1.1$ and $P=0.9$ (top and bottom panels in each figure, respectively). 
In all the simulations, we observe slowly decaying oscillations around
a different solitary wave near the initial perturbation, suggesting
that the singular solitary waves are stable in the time evolution of
the NLS equation (\ref{nls}). This agrees with the Lyapunov stability
of the solitary wave solutions which follow from the result of Theorem
\ref{theorem-main2}.

A closer inspection of the amplitude of the
wave points to a very slow (presumably power law) decay towards
a new solitary wave equilibrium. Additionally, in our (total) energy
conserving simulations, we observe a very weak exchange of energy
between the inner (head) and the outer regions of the solitary wave.
The latter may be also weakly affected by the approximate
nature of the numerical computations, e.g., by the numerical
approximation of the unit modulus at $x=\pm \ell_C$.

\begin{figure}[hbt]
	\includegraphics[scale=0.45]{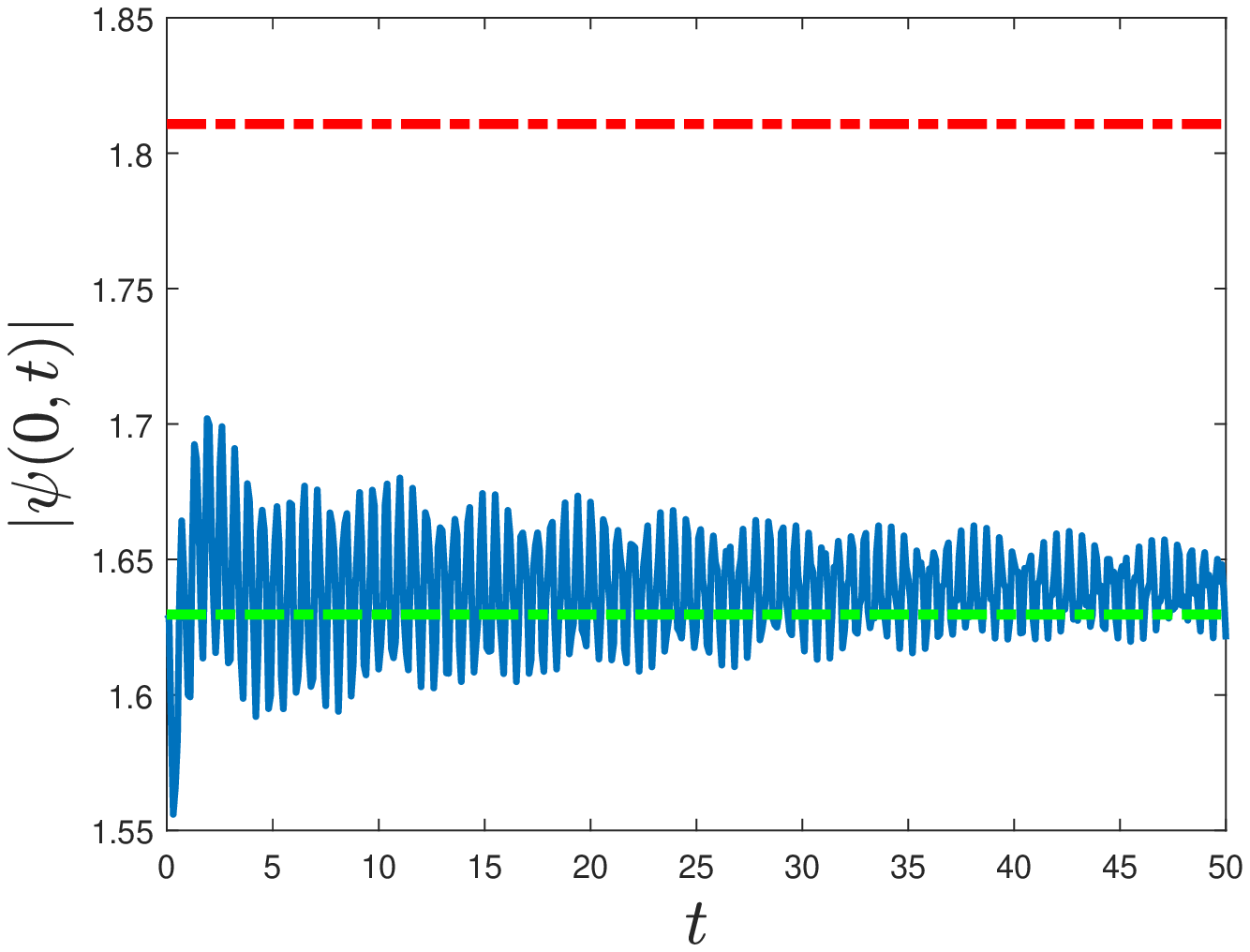}
	\includegraphics[scale=0.45]{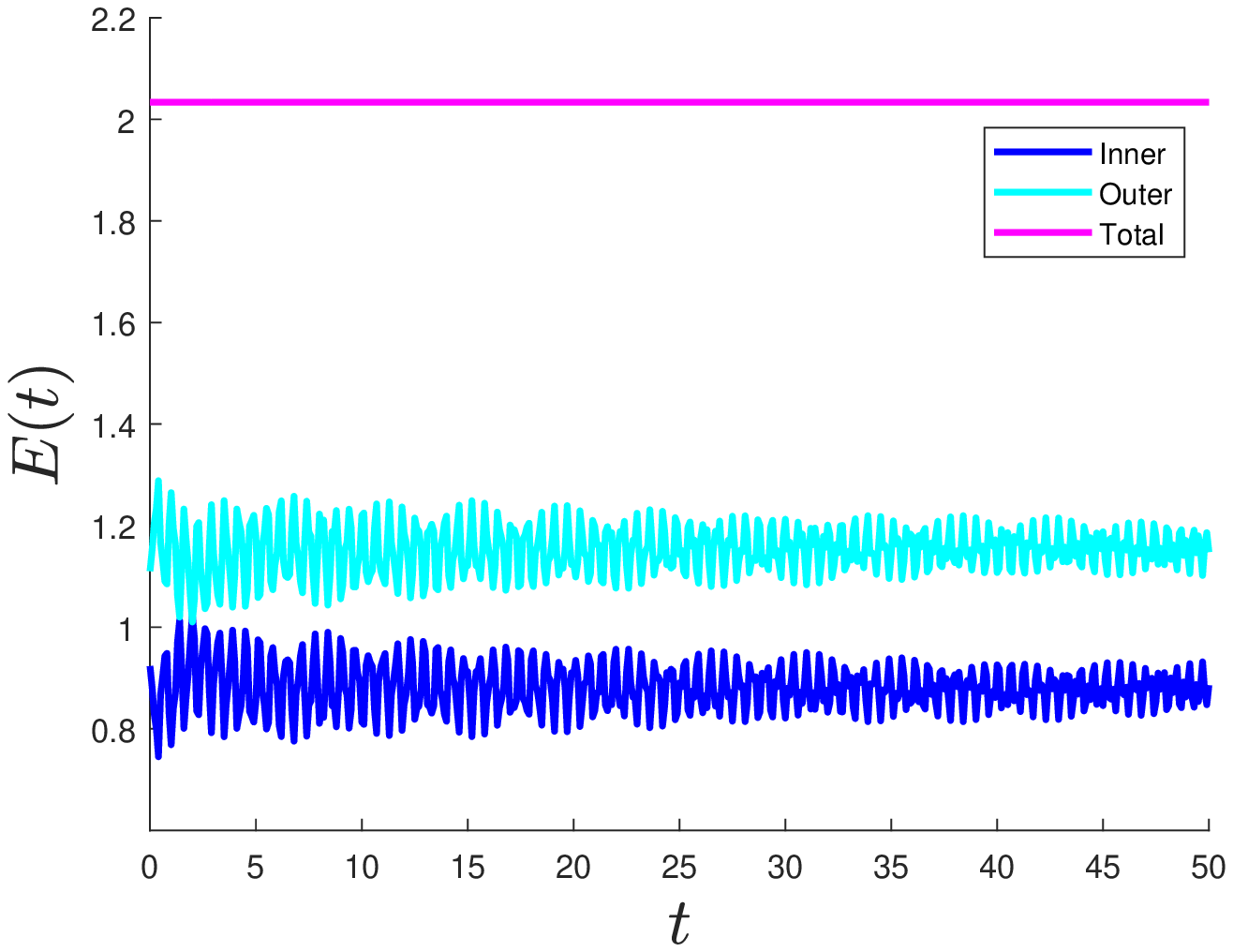} \\ 
	\includegraphics[scale=0.45]{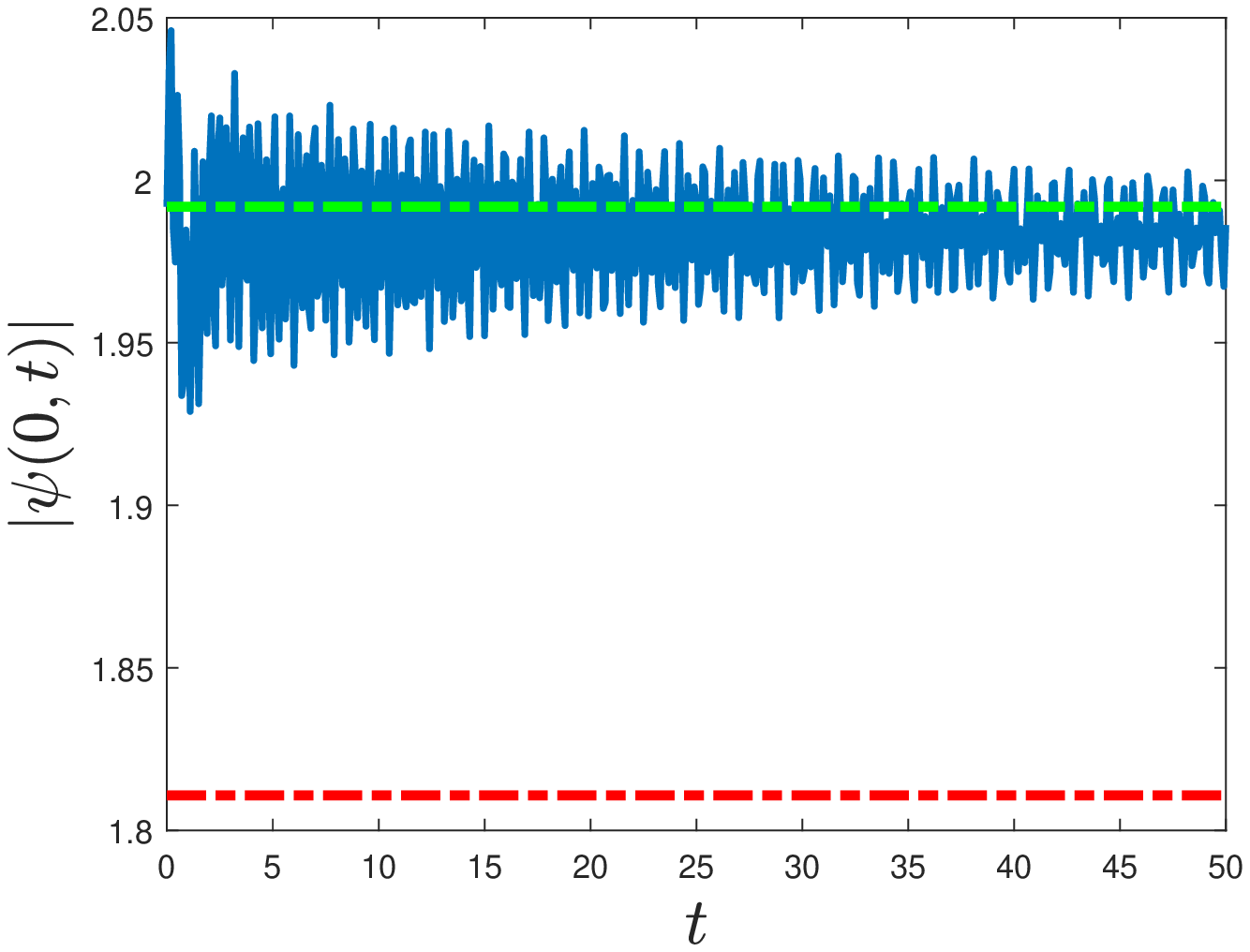}
	\includegraphics[scale=0.45]{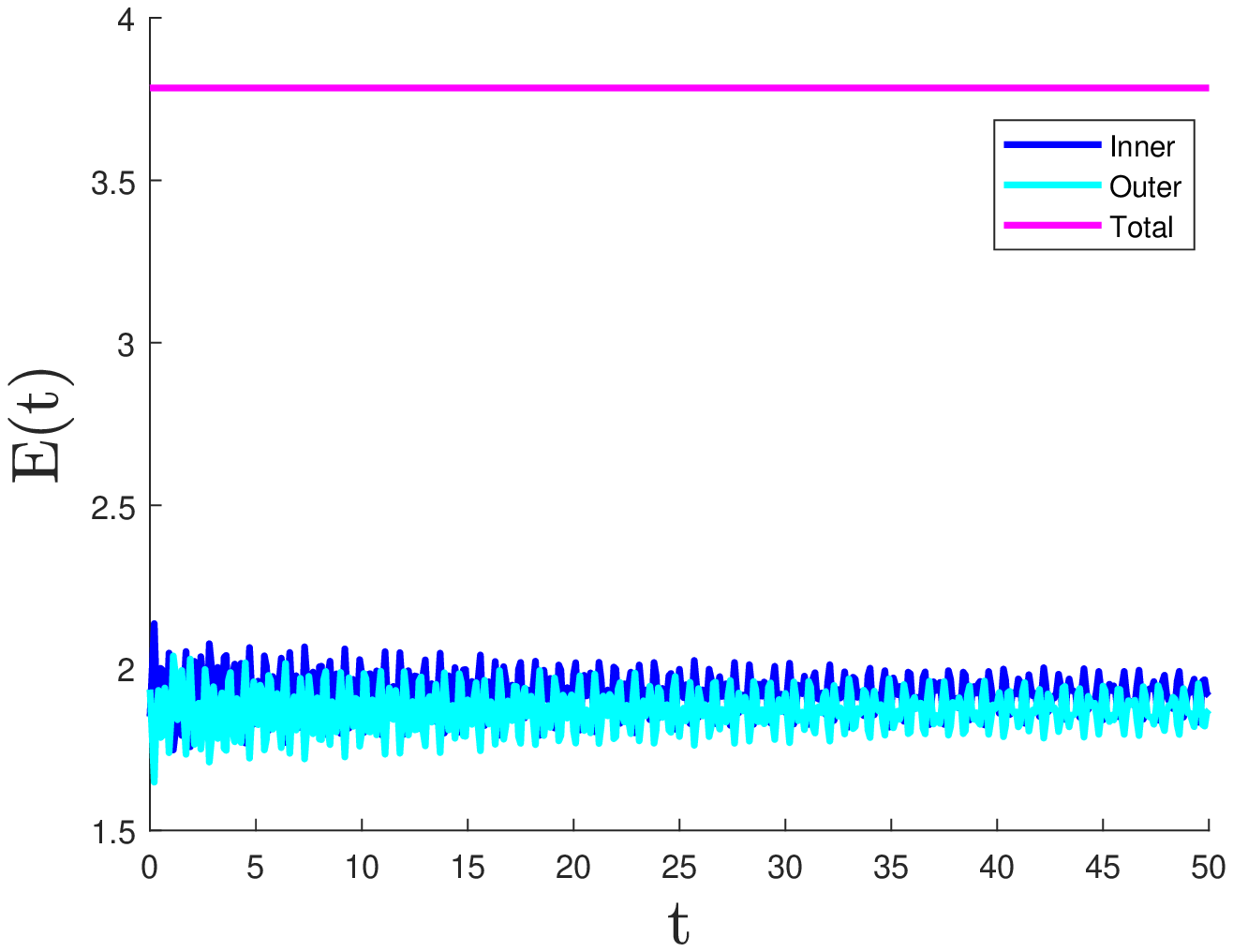}
	\caption{The same as Figure \ref{bell-pert} but for $C=0.5$.}
	\label{C=05-pert}
\end{figure}

\begin{figure}[hbt]
	\includegraphics[scale=0.45]{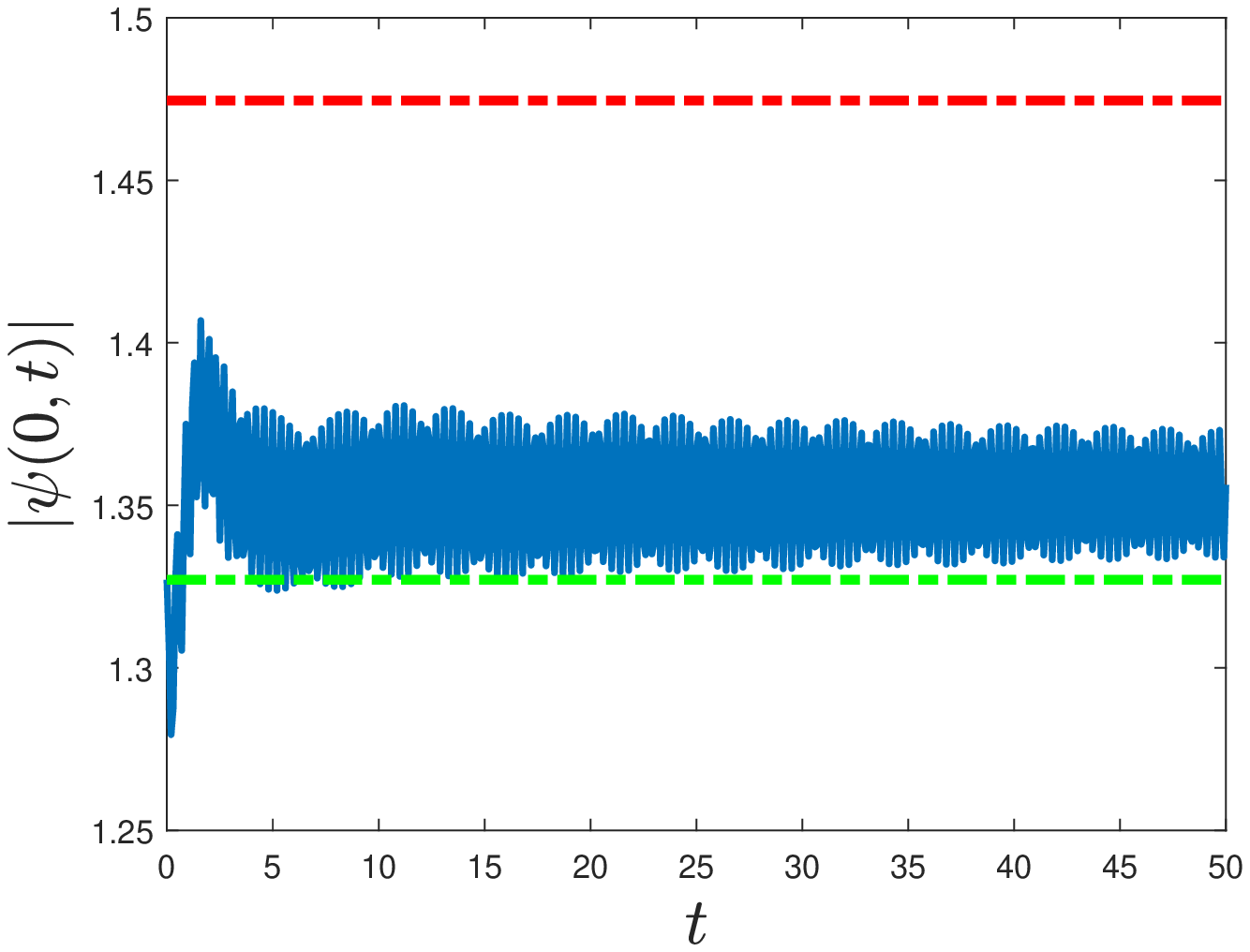}
	\includegraphics[scale=0.45]{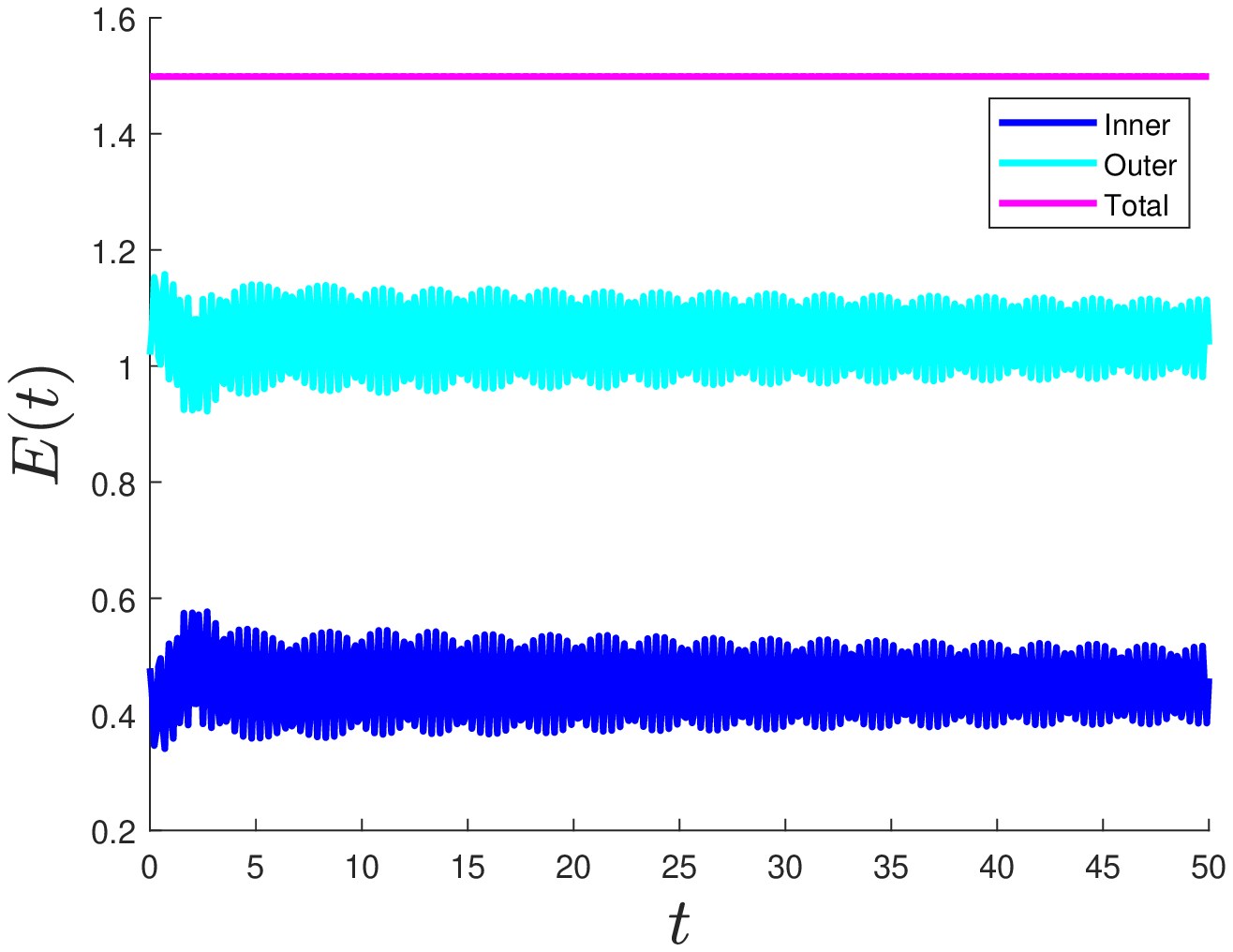}	\\
	\includegraphics[scale=0.45]{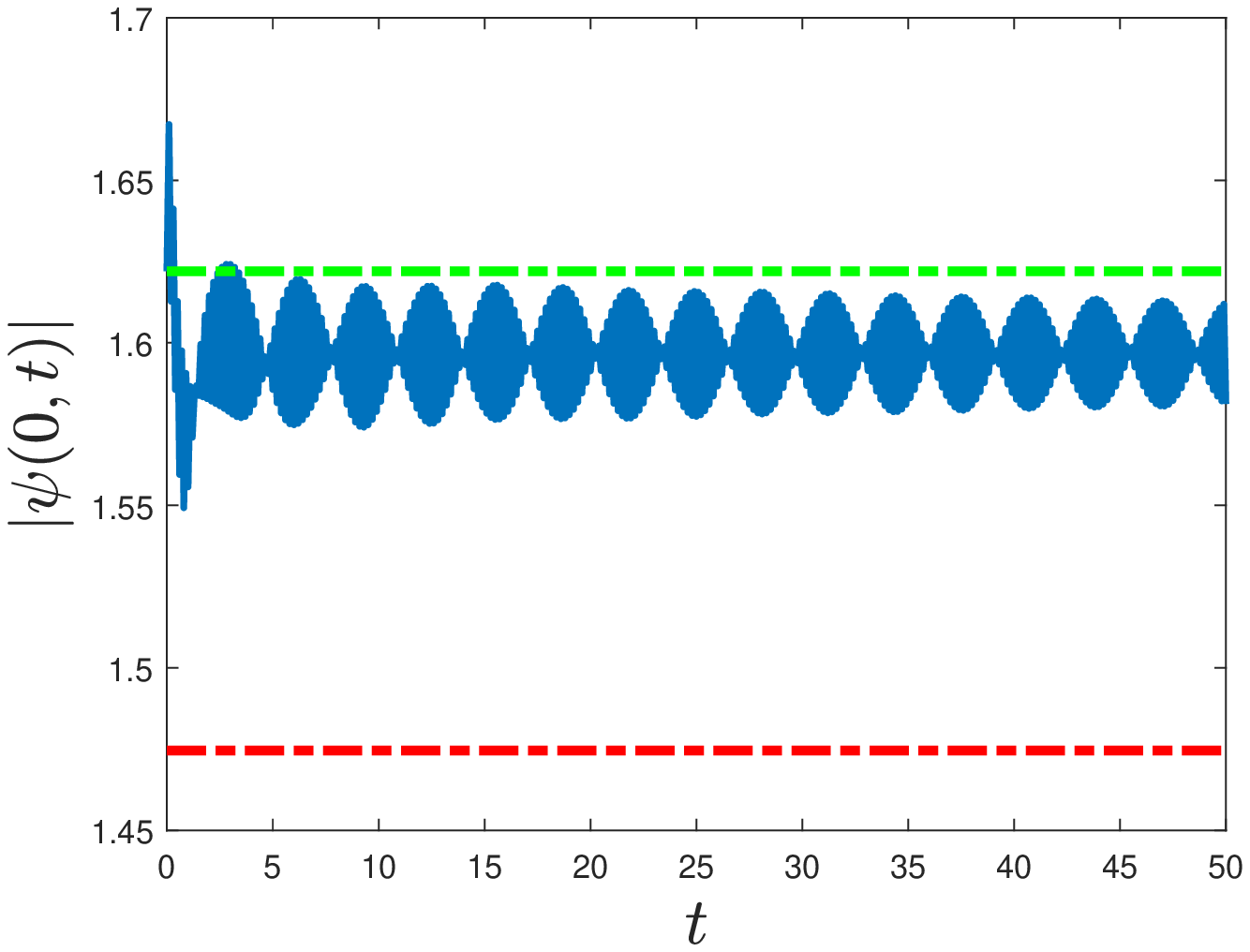}
	\includegraphics[scale=0.45]{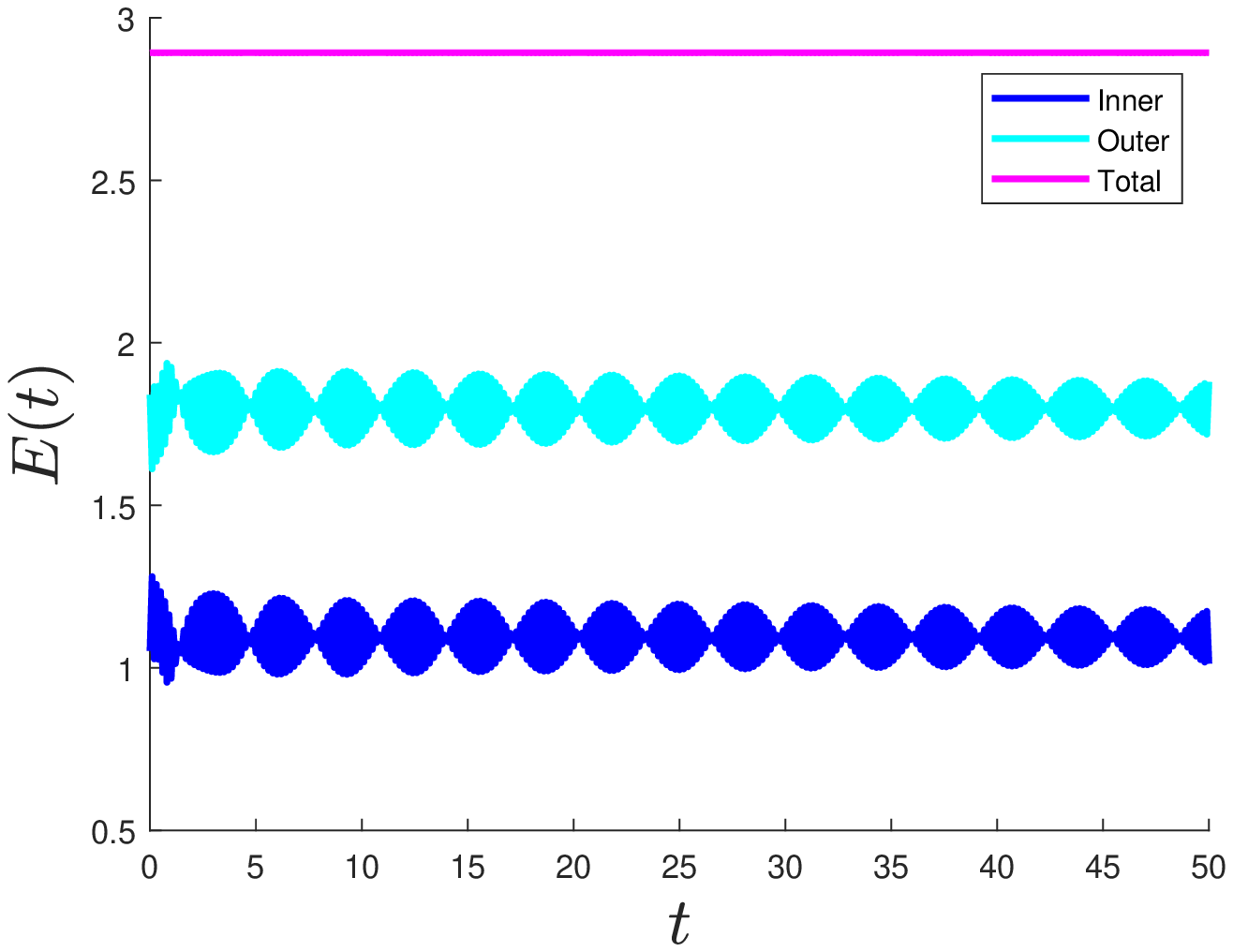}
	\caption{The same as in Figure \ref{bell-pert} but for $C=-0.5$.}
	\label{C=-05-pert}
\end{figure}

We remark that, due to numerical limitations, we are not able to investigate perturbations which change the singularity locations at $\pm L$ while keeping the same energy level $\mu$. In theory, the location of the singularity for these perturbations may change in time, because $\psi_{xx}$ is infinite when $|\psi|=1$, and the term $(1-|\psi|^2)\psi_{xx}$ in the NLS equation \eqref{nls} is indeterminate. In numerical simulations, however, the derivative $\psi_{xx}$ is replaced by a finite approximation, which results in the term $(1-|\psi|^2)\psi_{xx}$ being computed as 0 when $|\psi|=1$, in which case the NLS equation (\ref{nls}) implies that the singularity locations at $\pm L$ are preserved in the time evolution. 

\section{Conclusions}

In the present work, we have provided a variational characterization of solitary waves in a prototypical NLS model with intensity-dependent dispersion.
We have argued that minimization of mass at fixed energy and fixed length of the bell-shaped head is beneficial from an analytical point of view since it allows us to establish Lyapunov stability of the singular solitary waves. 
%More concretely, we have complemented our earlier work of~\cite{RKP} that detailed the nature and the asymptotics of the singular steady state solutions of the problem, with an additional relevant theorem. The latter confirms the existence of minimizers in this constrained variational problem and confirms their accordance (for given length of the ``head'' of the solution) with the steady states of~\cite{RKP}. 
This expected stability of the solitary waves was confirmed 
by direct dynamical simulations of the NLS model. We have observed in numerical simulations that perturbations of such waveforms lead to a slow relaxation 
of perturbed solitary waves to a new solitary wave within the family.

Among further open problems, we mention the rigorous analysis of well-posedness of the NLS model in the energy space where the solitary waves exist. 
It is also interesting to investigate how the singularity locations can change 
in the time evolution of the solitary waves, our analytical and numerical 
results rely on the fixed location of the singularities. Finally, 
it is interesting to study Lyapunov stability stability of other 
(sign-changing) solitary waves and periodic solutions discussed both in~\cite{OL2020} and~\cite{RKP}. It is also worth exploring 
generalizations of the NLS model in the settings of the discrete
(waveguide) systems, as well as in higher-dimensional systems. 
Such studies are deferred to future publications.

%\newpage

\end{document}